\documentclass[11pt]{article}
\usepackage[letterpaper,margin=1in]{geometry}
\usepackage[utf8x]{inputenc}
\usepackage[T1]{fontenc}
\usepackage{amsmath, amsthm, amssymb, bbm, enumerate,complexity}
\usepackage[ruled]{algorithm2e}
\usepackage{stmaryrd}
\usepackage{xcolor}
\usepackage[pagebackref,colorlinks,citecolor=blue,,linkcolor=magenta,bookmarks=true]{hyperref}
\usepackage[nameinlink,capitalize]{cleveref}

\title{Tight Correlation Bounds for Circuits Between AC0 and TC0}
\author{ Vinayak M. Kumar\thanks{Department of Computer Science, University of Texas at Austin. Email: {\tt vmkumar@utexas.edu}. Supported by NSF Grant CCF-2008076 and a Simons Investigator Award (\#409864, David Zuckerman). Part of this work was done while visiting Harvard University.}}
\date{}

\newcommand{\corr}{\func{corr}}

\newcommand{\dt}{\func{DT}}

\newcommand{\OR}{\func{OR}}

\newcommand{\tc}{\func{TC}}
\newcommand{\AND}{\func{AND}}
\newcommand{\NOT}{\func{NOT}}
\newcommand{\maj}{\func{MAJ}}

\newcommand{\any}{\func{ANY}}
\newcommand{\thr}{\func{THR}}
\newcommand{\ltf}{\func{LTF}}
\newcommand{\nc}{\func{NC}}
\newcommand{\gip}{\func{GIP}}
\newcommand{\rw}{\func{RW}}
\newcommand{\pargate}{\func{PAR}}
\newcommand{\gt}{\func{GC}}
\newcommand{\g}{\func{G}}

\newcommand{\ac}{\func{AC}^0}

\newcommand{\sgn}{\text{sgn}}
\newcommand{\ctr}{\text{ctr}}

\newcommand{\bits}{\{0,1\}}
\newcommand{\calF}{{\cal F}}
\newcommand{\calS}{{\cal S}}
\newcommand{\calC}{{\cal C}}
\newcommand{\calH}{{\cal H}}
\newcommand{\calE}{{\cal E}}

\renewcommand{\R}{\mathbb{R}}
\renewcommand{\E}{\mathbb{E}}

\newcommand{\N}{\mathbb{N}}

\newcommand{\eps}{\varepsilon}

\newcommand{\ceil}[1]{\lceil #1 \rceil}

\newtheorem{thm}{Theorem}[section]
\newtheorem*{remark}{Remark}

\newtheorem{cor}[thm]{Corollary}
\newtheorem{lem}[thm]{Lemma}
\newtheorem{defn}[thm]{Definition}
\newtheorem{claim}[thm]{Claim}
\begin{document}
\maketitle

\begin{abstract}

We initiate the study of generalized $\ac$ circuits comprised of arbitrary unbounded fan-in gates which only need to be constant over inputs of Hamming weight $\ge k$ (up to negations of the input bits), which we denote $\gt^0(k)$. The gate set of this class includes biased LTFs like the $k$-$\OR$ (outputs $1$ iff $\ge k$ bits are $1$) and $k$-$\AND$ (outputs $0$ iff $\ge k$ bits are $0$), and thus can be seen as an interpolation between $\ac$ and $\tc^0$. \ 

We establish a tight multi-switching lemma for $\gt^0(k)$ circuits, which bounds the probability that several depth-2 $\gt^0(k)$ circuits do not simultaneously simplify under a random restriction. We also establish a new depth reduction lemma such that coupled with our multi-switching lemma, we can show many results obtained from the multi-switching lemma for depth-$d$ size-$s$ $\ac$ circuits lifts to depth-$d$ size-$s^{.99}$ $\gt^0(.01\log s)$ circuits
with \emph{no loss in parameters} (other than hidden constants). 
Our result has the following applications:
\begin{itemize}
    \item Size-$2^{\Omega(n^{1/d})}$ depth-$d$ $\gt^0(\Omega(n^{1/d}))$ circuits do not correlate with parity (extending a result of H\aa stad (SICOMP, 2014)).
    \item Size-$n^{\Omega(\log n)}$ $\gt^0(\Omega(\log^2 n))$ circuits with $n^{.249}$ arbitrary threshold gates or $n^{.499}$ arbitrary symmetric gates exhibit exponentially small correlation against an explicit function (extending a result of Tan and Servedio (RANDOM, 2019)).
    \item  There is a seed length $O(\log^{d-1} m\log(m/\eps)\log\log(m))$ pseudorandom generator 
    against size-$m$ depth-$d$ $\gt^0(\log m)$ circuits, matching the $\ac$ lower bound of H\aa stad up to a $\log\log m$ factor (extending a result of Lyu (CCC, 2022)).
    \item Size-$m$ $\gt^0(\log m)$ circuits have exponentially small Fourier tails (extending a result of Tal (CCC, 2017)).
\end{itemize}
\clearpage 

\tableofcontents

\end{abstract}

\clearpage

\section{Introduction}
Proving superpolynomial circuit lower bounds against explicit functions is one of the most central questions in complexity theory. However, after the initial flurry of work resulting in Blum's lower bound of $3n-o(n)$  \cite{blum84}, followed by a recent revival 30 years later leading to the state of the art $3.1n-o(n)$ size lower bound by Li and Yang \cite{ly22}, this problem has proven to be extremely difficult. Furthermore, there are various proof barriers that give strong evidence that our current intuition is not developed enough to tackle this problem \cite{bjr75,rr94,aw09}.\ 

In order to gain more understanding on this problem, researchers considered circuits with constant depth whose gates are $\AND,\OR$, or $\func{NOT}$ with unbounded fanin. To this end there has been a fruitful line of work culminating in the state of the art average case hardness of depth $d$ size $2^{\Omega(n^{\frac{1}{d-1}})}$ $\ac$ circuits computing the parity and majority functions \cite{Hastad, imp12} (with this result being tight for parity). A natural followup question to ask is how powerful $\ac$ would then be if $\oplus$ (parity) or $\func{MAJ}$ (majority) gates were added, corresponding to the circuit classes $\func{AC}^0[\oplus]$ and $\func{TC}^0$. With regard to $\func{AC}^0[\oplus]$, Oliveira, Santhanam, and Srinivasan \cite{oss19} proved that it is average case hard for any 
size-$2^{\Omega(n^{1/(2d-4)})}$ $\ac_d[\oplus]$ circuit to compute $\func{MAJ}$, improving earlier work by Razborov \cite{R87}. 
Smolensky \cite{smo87} proved exponential size lower bounds even if one replaces the $\oplus$ gate with a $\func{MOD}_p$ gate for prime $p$ ($\func{MOD}_p$ is the gate that outputs 0 iff $p$ divides the sum of the input bits
). \ 

As we see, $\func{MAJ}$ is a hard function demonstrating exponential circuit lower bounds for almost all the circuit classes mentioned thus far, and so we would guess $\func{TC}^0$ is extremely powerful and thus challenging to show circuit lower bounds for. This is evident in the current state of the art for $\func{TC}^0$, which is stark in contrast with the landscape of $\func{AC}^0[\oplus]$. In 1993, Impagliazzo, Paturi, and Saks \cite{ips93} showed that parity is hard for depth-$d$ size-$\Omega(n^{1+\eps_{IPS}^{-d}})$ circuits for some constant $C_{IPS}>1$ 
, which remains as the current state of the art modulo the case of $d=2$, where Kane and Williams established a $n^{2.49}$-size lower bound \cite{kw10}. In fact, a bootstrapping result by Chen and Tell \cite{ct19} shows that if one slightly improves (e.g. decreases $C_{IPS}$) this superlinear lower bound against certain $\func{NC}^1$-complete problems ($\nc^1$ is the class of $O(\log n)$-depth, polysized, constant fan-in circuits), we would immediately get superpolynomial lower bounds and a separation of $\tc^0$ and $\func{NC}^1$, attesting to the hardness of this task.\

 Due to the halted state of affairs for $\tc^0$ circuits
 , we study a circuit class not as strong as $\func{TC}^0$, but still captures the motivation of analyzing ``$\ac$ with the power of majority.'' After it had been shown $\ac$ circuits cannot efficiently compute the majority of $n$ bits, it seemed natural that the next step would be to add unbounded $\maj$ gates to $\ac$ to create $\tc^0$. However, due to having unbounded fan-in, $\tc^0$ gives a size-$s$ circuit the power to calculate the majority of up to $s$ bits. Hence, one could argue the reason why size $s$ $\tc^0$ circuits are much harder to analyze than $\ac$ is because they are getting much more power than simply calculating the majority of $n$ bits when $s\gg n$. In order to maintain the unbounded fan-in property of the circuit but also ration the computational power we give $\ac$ to be ``just sufficient'' to compute the majority of $n$ bits, one can consider the following circuit class.
 
 \begin{defn}[$\ac(k)$ Circuits]
     Define the unbounded fan-in gates $k$-$\OR$ to output 1 iff there are at least $k$ ones in the input string, and $k$-$\AND$ to output 0 iff there are at least $k$ zeros in the input string. Define the class of constant depth circuits created by negations and $\{k'$-$\AND,k'$-$\OR\}$s for $k'\le k$ to be $\ac(k)$.
 \end{defn}
 One can observe that $\ac(n/2)$ is a natural circuit class that contains the majority of $n$ bits and doesn't add ``extra power'' 
 like the majority of a much larger quantity of bits. Therefore, analyzing $\ac(n/2)$ will give us a better understanding on how much power majority gives to circuits. More generally, $\ac(k)$ also allows us to nicely interpolate between $\ac$ and $\tc^0$, since a size-$s$ $\ac$ circuit is characterized by $\ac(1)$, while a size-$s$ $\tc^0$ circuit is characterized by $\ac(s/2)$. Hence, 
 studying $\ac(k)$ for increasing $k$ is a necessary step and a compelling intermediary model that can help us understand the power of $\tc^0$.

For how large of a $k$ will $\ac(k)$ trivially collapse to $\ac$? An immediate observation is that $\ac(k)$ contains the majority gate over $2k$ bits, for which we know $2^{\Omega(k^{1/2d})}$-size $\ac$ lower bounds. Hence, for $k=\polylog(n)$, we have a superpolynomial size seperation between $\ac$ and $\ac(k)$. Even for any $k=\omega(1)$, it is unknown whether $\ac(k)$ is equivalent to $\ac$. A standard argument would be to represent a $k$-$\OR$ with fan-in $m$ as a width-$k$ DNF with $\binom{m}k$ clauses (check over all size $k$ subsets of input bits to see if some subset are all $1$s) or a width-$(m-k)$ CNF with $\binom{m}k$ clauses (check over all size $m-k$ subsets of input bits to see if all subsets contain some $1$). Therefore, if we have a size $s$ circuit made from $k$-$\OR$ and $k$-$\AND$ gates, we can turn this into an $\ac$ circuit with size $s\cdot \binom{s}k\approx s^k$ (since a gate in the original circuit can have fan in size up to $s$) and depth $d+1$ (one can naively get depth $2d$, but by alternating CNFs and DNFs, we can collapse the depth to $d+1$). Hence, we see that a size $s$ lower bound for $\ac_d$ translates to a size $s^{1/k}$ lower bound for $\ac_{d-1}(k)$. This reduction is not an equivalence, as we pay with a reduction in depth, as well as an asymptotically weaker size lower bound for any $k = \omega(1)$. For example, a polynomial size bound for $\ac$ cannot be converted to a polynomial size lower bound for $\ac(k)$ for any superconstant $k$. Consequently, the relationship between $\ac$ and $\ac(k)$ already becomes nontrivial in the mild regime of $k=\omega(1)$.\

In this paper, we study an \emph{even more} general class of circuits, which we denote as $\gt^0(k)$. 

\begin{defn}[$\g(k)$ gates/$\gt^0(k)$ circuits] Let $\g(k)$ be the set of all unbounded fan-in gates that are constant over all input bits with $\ge k$ ones, or over all input bits with $\ge k$ zeros (notice for $k\ge 1$ this includes negations by definition). Define $\gt^0(k)$ to be the class of constant depth circuits created from $\g(k)$ gates.
\end{defn}

Some concrete examples of $\g(k)$ gates are arbitrary gates of fan-in $k$, majority of $2k$ bits, the $k$-$\OR$, and functions that compute parity if the input has $<k$ ones, and is $0$ otherwise. Notice that this is indeed a generalization of $\ac(k)$.

On top of being an alternative generalization of $\AND/\OR$ gates which may be of independent interest, one nice property about $\g(k)$ is that it includes a generalized notion of $k$-$\AND$ and $k$-$\OR$ gates to arbitrary $\ltf$s (functions of the form $\text{sgn}(\sum w_i x_i - \theta)$).

\begin{defn}[$k$-balanced $\ltf$s]
\label{defn:bal}
    Let $f(x) = \sgn(\sum_{i=1}^n w_i x_i - \theta)$ be an $\ltf$, and let $\sigma: [n]\to[n]$ be a permutation sorting the $w_i$ in increasing magnitude (i.e. $|w_{\sigma(1)}|\le \dots\le |w_{\sigma(n)}|$). We say $f$ is \emph{$k$-balanced} if $k$ is the smallest index $j$ such that $-\sum_{i\le j}|w_{\sigma(i)}|+\sum_{i > j}|w_{\sigma(i)}| < |\theta|$.
\end{defn}
One can verify that $k$-balanced $\ltf$s are indeed in $\g(k)$ (see \Cref{thm:thrgk}). Therefore, our results can also be seen as a study of arbitrary $\ltf$s that are biased.\

Various notions of \emph{balancedness} (or \emph{regularity} in some literature) for $\ltf$s has been defined in previous work about threshold functions \cite{Ser06, pot19, hhtt21}, but are all distinct from the combinatorial definition we have proposed. In light of being able to show lower bounds for this characterization of \emph{balanced}, it may be of interest to explore this class of balanced $\ltf$s in other contexts regarding $\ltf$ circuit complexity. \

\subsection{Our Results}

We outline all the results we obtain regarding $\ac(k)$ (or more generally $\gt^0(k)$) circuits. The core result from which all the other results are derived from is an optimal multi-switching lemma for $\gt^0(k)$ circuits. We state the result without getting into the fine-grained definitions.

\begin{thm}[Multi-Switching Lemma for $\gt{(k)}$ Circuits (Informal)]
\label{thm:informalmulti}
Let ${\cal F} = \{F_1,\dots, F_m\}$ be a list of 
$\g(k)\circ \AND_w$ circuits on $\{0,1\}^n$. Then \[ \Pr_{\rho\sim R_p}[{\cal F}|_\rho\text{ do not all simultaneously ``simplify''}]\le (2^km)^{t/r}(O(pw))^t \]
\end{thm}

The theorem statement and proof is formally written in \Cref{thm:switch}. This bound can be proven to be optimal in the regime of large $t$. See the appendix (\Cref{apx:tight}) for the proofs of this claim.

It is illuminating to compare this result to the multi-switching lemma for $\ac$ circuits, which bounds the probability by the very similar expression of $m^{t/r}(O(pw))^t$. The only difference is that in the new lemma, the $m$ and $2^k$ are coupled in the base of the exponent. This seems to hint that as long as $2^k = O(m)$, one gets the same probability bound when using either the $\ac$ or $\gt^0(k)$ version of the multi-switching lemma. In practice, the parameter $m$ is upper bounded by $s$, the size of the circuit. Hence, we would intuit that any result obtained from the multi-switching lemma for depth $d$ size $s$ $\ac$ circuits can then be lifted to size $s$ $\gt_d^0(\log s)$ circuits. This indeed turns out to be the case as we demonstrate through four different results. We obtain a surprising lifting theorem: any depth $d$ size $s$ $\ac$ lower bound obtained by the multi-switching lemma immediately lifts to depth $d$ size $s^{.99}$ $\gt^0(.01\log s)$-circuits \emph{with no loss in parameters}. We demonstrate three different results exhibiting this phenomenon. \ 

For the first result, denote $\pargate$ to be the parity gate.

\begin{thm}[Optimal Correlation Bounds Against Parity] Let $C$ be a size $m$ depth $d$ $\gt^0(k)$-circuit. Then the correlation of $C$ against parity is \[|\E_x[(-1)^{C(x) + \pargate(x)}]|\le 2^{-\Omega_d(n/(k + \log m)^{d-1}) + k}.\]
\end{thm}

In particular, we get a $2^{\Omega(n^{1/d})}$-size lower bound for $\gt^0_d(\Omega(n^{1/d}))$ circuits almost matching the lower bound of $2^{\Omega(n^{1/(d-1)})}$ we know for $\ac$! This is especially surprising in light of the fact that $\gt^0_d(\Omega(n^{1/d}))$ is a much stronger class than $\ac$; there exist singleton $\g(n^{1/d})$ gates that cannot be computed by size $O(2^{n^{1/2d}})$ $\ac_d$ circuits. This can be seen as a limited dual result to \cite{R87}, who showed $\ac_d$ augmented with parity gates requires size $2^{\Omega(n^{1/2d})}$ to compute majority, whereas we show $\ac_d$ augmented with $n^{1/d}$-biased majority gates requires size $2^{\Omega(n^{1/d})}$ to compute parity. It also contrasts with \cite{oss19}, who surprisingly showed that adding parity gates to $\ac$ improved optimal circuit constructions of majority. Here, we show that majority gates whose threshold value is shifted to $\Omega(n^{1/d})$ has no effect on $\ac$'s ability to calculate parity, even though such gates adds a lot of power to $\ac$. (majority gates whose threshold has been biased to $n^{1/d}$ cannot be computed by size $2^{\Omega(n^{1/2d^2})}$ $\ac_d$ circuits). \

Notice that this result is tight in an extremely sensitive way. Letting $\pargate_n$ denote the parity gate over $n$ bits, we see $\pargate_{n^{1/d}}\in \g(n^{1/d})$, and we can calculate the parity of $n$ bits by creating a depth $d$ $n^{1/d}$-ary tree of $\pargate_{n^{1/d}}$ gates, where the $i$th layer from the bottom has $n^{1-i/d}$ $\pargate_{n^{1/d}}$ gates that take the parity of all the bits fed below it in blocks of $n^{1/d}$. This is a depth $d$ size $O(n^{1-1/d})$ circuit computing parity. Therefore, we have a simple counterexample of a $\gt^0_d(n^{1/d})$ circuit computing parity (which is sublinear in size!). This demonstrates a sharp threshold behavior where the exponential lower bound of $2^{\Omega(n^{1/d})}$ is tight up to the hidden constant factor of the $\Omega(\cdot)$ in $\gt_d^0(\Omega(n^{1/d}))$, and if the constant is too large, we suddenly go from requiring exponentially large circuits to only needing sublinear size ones.\ 

This theorem is tight in all other parameters as well. We show that this result is tight in the size parameter by giving a size-$2^{\Omega(n^{1/d})}$ $\gt^0(.1n^{1/d})$ circuit computing parity. Furthermore, we show that the correlation bound is tight by giving a size-$m$ $\gt^0(k)$ circuit that approximates parity.

For what $k$ will analyzing $\ac(k)$ give implications for $\tc^0$? A result by Allender and Kouck\'{y} (\cite{ak10}, Theorem 3.8) states that there exists an absolute constant $C_{AK}$ such that $\maj_n$ can be written as an $\ac(n^\eps)$ circuit with depth $\le C_{AK}/\eps$ and size $O(n^{1+\eps})$. Therefore, beating the current state of the art depth $d$ size $\Omega(n^{1+C_{IPS}^{-d}})$ lower bound for $\tc^0$ reduces to beating depth $C_{AK}d/\eps$ size $n^{(1+50^{-d})(2+\eps)}$ lower bounds for $\gt^0(n^\eps)$ circuits for any choice of $\eps$. In our paper, we show exponential size lower bounds against parity when $\eps = 1/d$ but at depth $d$ rather than $C_{AK}d^2$. It would be interesting to see whether with some $\nc^1$-complete problem can display strong lower bounds for $\ac(n^{1/d})$ for depth larger than $d$, even if it may be less than $C_{AK}d^2$ (but a function other than parity would need to be considered).\

Another angle researchers have taken towards understanding the power of threshold circuits has been to  start with an $\ac$ circuit and augment some of the gates to arbitrary threshold gates \cite{v07,ls11,st19}. Our multi-switching lemma shows that we can instead start with a base $\gt^0(\log s)$ circuit and obtain the same state of the art parameters as \cite{st19} if we started from an $\ac$ circuit.

\begin{thm}
    There exists a function $\rw\in \func{P}$ (introduced by Razborov and Widgerson \cite{rw93}) and absolute constant $\tau$ such that for $C$, a size $n^{\Omega(\log n)}$ $\gt^0(\Omega(\log^2 n))$-circuit with $n^{.249}$ $\thr$ gates, we have \[ |\E_x[(-1)^{\rw(x) + C(x)}]|\le 2^{-\tilde{\Omega}(n^{.249})}.\]
\end{thm}

The original motivation to study $\ac$ with a small number of $\thr$ gates was to use this to gradually convert circuits gate by gate from $\ac$ to $\tc^0$. This result ``speeds up'' this process by augmenting \emph{all} $\ac$ gates to $\g(\log^2 n)$ gates (which contain unbalanced $\ltf$s as discussed above). If one tried proving this theorem by expanding the $\gt^0(\log^2 n)$ circuits into an $\ac$ circuit, completing the proof would require solving a longstanding open problem regarding correlation bounds against $\omega(\log n)$-party NOF protocols! In \Cref{sec:arbgate}, we point out this observation explicitly along with the proof. \ 

As another application, we can create PRGs for $\gt^0(\log m)$ circuits whose seed length matches that of size $m$ $\ac$ circuits. This is accomplished by fully derandomizing \Cref{thm:informalmulti} and using the partition-based PRG approach in \cite{Lyu}. The resulting PRG for $\gt_d^0(\log m)$ has identical seed length as Lyu's PRG, thereby also matching H\aa stad's $\ac$ lower bound barrier up to a $\log\log m$ factor (see \cite{tx13,st19b,kel21} for a discussion on why an $o(\log^d(m/\eps))$ seed length implies breakthrough circuit lower bounds).

\begin{thm}
    For every $m,n,d\ge 3$ and $\eps > 0$, there is an $\eps$-PRG for size-$m$ $\gt_d^0(\log m)$ with seed length $O(\log^{d-1}(m)\log(m/\eps)\log\log m)$
\end{thm}

The proof is covered in \Cref{sec:prg}. Notice that if we had simply expanded out all gates as width $\log m$ CNF/DNFs, we would have a size $\approx m^{\log m}$ $\ac_{d+1}$ circuit, and plugging in Lyu's near-optimal PRG would yield us a suboptimal seed length of $O((\log^2m + \log(1/\eps))\log^{2d}m\log\log m)$. 

Finally, we establish results on the Fourier spectrum of $\gt^0(k)$ circuits. It can be shown that every Boolean function, when written as a map $\{\pm 1\}^n\to\{\pm 1\}$, can be uniquely expressed as a multivariate polynomial $f(x) = \sum_{S\subset[n]}\widehat{f}(S)\prod_{i\in S}x_i$. We show exponentially small Fourier tail bounds for any $C\in \gt^0(k)$. More concretely,

\begin{thm}
    For arbitrary $C\in\gt_d^0(k)$ of size $m$, the following is true for any $0\le \ell\le n$. \[\sum_{|S|\ge \ell} \widehat{C}(S)^2\le 2^{-\Omega\left(\frac{\ell}{(k+\log m)^{d-1}}\right)+k}\]
\end{thm}

Linial, Kushilevitz, Mansour, and Tal \cite{km91,lmn93,man92,tal17} showed that with small Fourier tails, one can get a variety of Fourier structure results, efficient learning algorithms, and correlation bounds. We demonstrate applications of such techniques to $\gt^0(k)$ in detail in \Cref{sec:fourier}. 

\section{Overview of the Proof of $\gt^0(k)$ Lower Bounds}

The novel ideas of this paper can be observed in the switching lemma and its use in collapsing the depth of $\gt^0(k)$ circuits. Therefore, we give an outline of this portion of the argument. Getting the specialized theorems noted above is then a matter of applying the lemma in various settings.

In order to prove correlation bounds against $\gt^0(k)$ cirucits, we use the framework of showing that such circuits simplify under \emph{random restrictions}. For simplicity, one can think of a random restriction as a partial assignment created by keeping each variable unfixed with some probability $p$, and fixed to $0$ or $1$ with $\frac{1-p}2$ probability each. Showing constant depth circuits simplify under random restrictions is usually done via following two steps: 
\begin{itemize}
\item Establish a switching/multi-switching lemma that states all the bottom depth-2 subcircuits ``simplify'' with great probability after a random restriction
\item Establish a depth reduction lemma which states when the bottom subcircuits are ``simplified'', one can effectively reduce the depth of the circuit, and induct.
\end{itemize}

If both of these are established, one can argue the parity function uncorrelates with these constant depth circuits, since parity reduces to a parity on a subset of bits when acted on by a random restriction, while the circuit simplifies to a constant with high probability. This approach is used to prove the well known lower bounds for $\ac$ circuits. The first bullet point is known to be the technical meat of the proof for $\ac$, requiring complex encoding, witness, or inductive arguments to show that the subcircuits simplify to decision trees. The second bullet point is swept under the rug since depth reduction is almost immediate in the $\ac$ case. \ 

In this paper, we will extend both bullet points to be applicable to $\gt^0(k)$ circuits. To prove the switching lemma, we extend Razborov's encoding/witness argument, intuitively showing if our subcircuit's top gate was a $G(k)$ gate instead of an $\OR$, we only need to store $k$ more bits in our encoding/witness (this will be described more in the next paragraph). To prove a depth reduction lemma, we will require a more involved proof that uses the full power of decision trees. In the $\ac$ argument, only the fact that a decision tree could be unraveled into a small width DNF or CNF was used. Here, using an additional structure property of these DNF/CNF stemming from its derivation from a decision tree, we can extend this depth reduction argument to $\gt^0(k)$. Such a depth reduction wouldn't be possible otherwise, and essentially uses the additional structure of decision trees.\

\subsection{The Switching Lemma}
\label{sub:switch}
To grasp this section, an understanding of Lyu's witness/transcript proof of the switching lemma \cite{Lyu} would be helpful. We still give an overview of the proof here and provide intuition from an information-theoretic lens, which differs in certain places than the intuition presented in \cite{Lyu}. After the overview, we will highlight the necessary changes needed to prove the more general lemma for $\gt^0(k)$ circuits.\

Say we have a $k$-$\OR\circ \AND$ circuit $F$ (in the formal proof, we consider general $\g(k)\circ \{\AND,\OR\}$ circuits). For $\Lambda\subset[n]$ and $z\in\bits^n$, denote $\rho(\Lambda, z)$ to be the restriction/partial assignment where all variables whose indices are in $\Lambda$ are kept alive/unfixed, and all remaining variables $x_i$ with $i\notin \Lambda$ are fixed to the corresponding bit in $z$, $z_i$. Consider a random restriction $\rho(\Lambda, z)$, where $z\sim \{0,1\}^n$ is a uniformly random ground assignment, and $\Lambda$ is a random subset of $[n]$ such that each element is added with probability $p$. To show that with low probability, $F|_\rho$ has decision tree depth $\ge t$, it suffices to create a specific \emph{canonical decision tree} (CDT) for each $\rho$ and argue that this tree has depth $\ge t$ with low probability (because if the decision tree depth of $F|_\rho$ is $\ge t$, then surely the canonical decision tree has depth $\ge t$). We consider the following CDT, where we first initialize a counter $ctr\gets 0$, and then scan the bottom layer clauses from left to right.

\begin{itemize} 
\item If the clause is fixed to $1$ and $ctr = k$, terminate since we know that $F$ evaluates to 1. Otherwise increment $ctr$ and move to the next clause.
\item If the clause is fixed to $0$, move to the next clause.
\item If the clause is ambiguous, query all variables in the clause, and behave accordingly as above.
\end{itemize}

If we think of our CDT as an algorithm that queries certain bits of the input, then \emph{bad} $\rho$ that creates a depth $\ge t$ CDT will produce a unique ``transcript'' of large size recording the behavior of CDT (i.e. the clauses and variables the CDT queries from). Like \cite{Lyu}, we consider transcripts that store $(\ell_i)$, the indices of the clauses queried, along with a set $P$ that further elucidates which variables in the clauses were queried in an information-efficient manner. We get the following inequality \begin{align}\Pr_{\rho} [\dt(F|_{\rho})\ge t] & \le \Pr_{\rho} [\func{CDT}(F|_{\rho})\ge t] \nonumber\\ &  \le \sum_{\substack{\text{large transcripts} \\ (\ell_i, P)}}\Pr_{\rho}[(\ell_i, P)\text{ is a large transcript for }\rho]\label{eqn:1}\end{align} via the union bound. A natural thought is to then bound each term in the sum. Unfortunately, it turns out that the number of transcripts $(\ell_i, P)$, when counted naively by multiplying the total possible lists $(\ell_i)$ by the total possible sets $P$, is far too large to get our switching lemma due to the vast amount of possible $(\ell_i)$. \ 

However, it turns out that $(\ell_i)$ contains \emph{redundant} information. Say $P$ is a partial transcript for $\rho$ if it can be completed with a suitable $(\ell_i)$ to form a transcript for $\rho$. We can show that given $\rho$ and $P$ that is a partial transcript for it, there is a unique list $(\ell_i)$ that completes $P$ to a full transcript. Hence \begin{align}&\sum_{\substack{\text{large transcripts} \\ (\ell_i, P)}}\Pr_{\rho}[(\ell_i, P)\text{ is a large transcript for }\rho] \nonumber\\ & = \sum_{\text{partial transcripts }P}\Pr_{\rho}[P\text{ is a partial transcript for }\rho]\label{eqn:2}\end{align} which is a sum of far fewer terms, making the union bound feasible. It remains to bound each individual term in the sum. \ 

We want to bound the probability a particular $P$ is a partial transcript for $\rho$. If we were given the complementary $(\ell_i)$'s, this would be easy. The $(\ell_i)$ along with $P$ would give a transcript of the specific set of $\ge t$ variables that the CDT queried, which $\rho$ must keep alive in order to have any hope of $(\ell_i,P)$ being a transcript for $\rho$. This would happen with probability $\le p^t$, which is a sufficiently small probability to apply the union bound. However, the trickiness arises due to $\ell_i$ not being specified. It turns out different $(\ell_i)$ might couple with the same $P$ to form transcripts for different $\rho$! Therefore if we use no information about $\rho$, then we have no hope of recovering a fixed $(\ell_i)$.\

On the other hand, if we were given complete information about $\rho$, then we can recover a unique $(\ell_i)$ or deduce none exists. However, this eliminates all randomness of $\rho$ and we get the trivial large upper bound of $1$ for each term. Therefore, for such an approach to work, we need to condition on partial information about $\rho$ and hope that it is enough information to recover $(\ell_i)$ but not too much information to the point where we get a weak bound on the probability due to the lack of randomness. \ 

This motivates us to think of a restriction by first assigning a uniform random string $z$ to $x$ and then covering up a $p$-subset $\Lambda$ with stars to create a restriction $\rho(\Lambda, z)$. The intuition for this is that hopefully the random string $z$, combined with $P$, will be enough information from $\rho$ to fix $(\ell_i)$, from which we can use the remaining randomness in $\rho$ (namely $\Lambda$) to obtain the $p^t$ bound. In particular, we hope that there is a ``transcript searcher'' ${\cal S}$, which on input $(z, P)$, can recover a completed transcript $(\ell_i, P)$ such that all $\rho$ designed by initially assigning $x=z$ will have partial transcript $P$ only if $(\ell_i,P)$ is its transcript. If such a function exists, then we could say

\begin{align*}   \Pr_\rho[P\text{ is a partial transcript for }\rho] & = \E_{z\sim U_n}\Pr_{\Lambda}[P\text{ is a partial transcript for $\rho(\Lambda, z)$}] \\ &  = \E_{z\sim U_n}\Pr_{\Lambda}[{\cal S}(z, P)\text{ is a transcript for $\rho(\Lambda,z)$}] \\ &  \le p^t \end{align*} where the last inequality follows since $\rho$ must keep the variables in the transcript alive. Alas, such an ${\cal S}$ cannot exist. There can exist different restrictions created from the same ground assignment $z$ that are witnessed by different completions of $P$ (this ambiguity is an unavoidable side effect of not being able to condition on all information about $\rho$). We cannot hope for a unique completion, but what if our ${\cal S}$ output all of these potential completions with decent probability over the randomness in $z$? Say $\rho$ is \emph{good} if $P$ is a partial transcript for it. In formal terms, say we can construct ${\cal S}$ such that for any good $\rho$, $\Pr_z [{\cal S}(z,P)\text{ is a partial transcript for }\rho]\ge \gamma$ (earlier we were demanding $\gamma = 1$, which turned out to be impossible). Then we can deduce \begin{align}   \Pr_\rho[P\text{ is a partial transcript for }\rho] & = \E_{\Lambda}\Pr_{z}[\text{$\rho(\Lambda, z)$ is good}] \nonumber\\ &  = \E_{\Lambda}\frac{\Pr_{z}[{\cal S}(z, P)\text{ is a transcript for $\rho(\Lambda,z)$}]}{\Pr_{z}[{\cal S}(z, P)\text{ is a transcript for $\rho(\Lambda,z)$} | \text{$\rho(\Lambda, z)$ is good}]} \nonumber\\ & \le \frac{1}\gamma \E_\Lambda \Pr_{z}[{\cal S}(z, P)\text{ is a transcript for $\rho(\Lambda,z)$} \nonumber\\ & = \frac{1}\gamma \E_z \Pr_{\Lambda}[{\cal S}(z, P)\text{ is a transcript for $\rho(\Lambda,z)$}]\nonumber\\ & \le p^t/\gamma.\label{eqn:3} \end{align} 

Stringing Equations (\ref{eqn:1}),(\ref{eqn:2}), and (\ref{eqn:3}) lets us bound \[\Pr_{\rho} [\dt(F|_{\rho})\ge t]\le (p^t/\gamma)\cdot \#\{\text{partial transcripts }P\}\] 
It turns out we can define our partial transcripts $P$ and construct a transcript searcher such that the above term is small enough to give us the desired switching lemma. See \Cref{thm:switch} for the technical details.



\subsubsection{Comparison to \cite{Lyu}}
\label{subsub:comparison}
Although the proof structure for proving our switching lemma is similar to Lyu's \cite{Lyu} proof of the $\ac$ switching lemma, some changes are necessary to accommodate the general structure of $\g(k)\circ\{\AND,\OR\}$ circuits.

\begin{itemize}
\item We need to create a more complex CDT that can compute $\g(k)\circ\{\AND,\OR\}$ circuits, and a corresponding new definition of witnesses/partial witnesses that records the transcript of the complex CDT so that our witness searcher can effectively reconstruct a transcript given information about $\rho$ and a partial witness.

\item Because our CDT contains more steps, there will naturally be more possible transcripts/witnesses. As the switching lemma hinges on a low quantity of possible partial witnesses to union bound over, we need to argue with our new CDT, the number of partial witnesses can be controlled by the parameter $k$. This makes designing the CDT and partial witness to be an act of balancing contrasting parameters

\begin{itemize} 

\item For example, the more complicated a CDT procedure is, the closer to the true decision tree depth it will reach (and hence a tighter bound on $\Pr_\rho[\func{CDT}(F|_\rho)\ge t]$ can be expected), but the larger the possible number of transcripts it will have (thereby increasing the number of terms we union bound over). Therefore, this approach demands the designed CDT to be complicated enough to give a small depth decision tree with high probability, but simple enough to be tractable to analyze with a union bound.

\item Similarly, the more that a partial witness keeps track of, the larger amount of possible partial witnesses we will need to union bound over. However, if we keep track of too little, there will not exist an effective witness searcher that can use the information from the partial witness to construct the whole witness. Hence we need to keep track of just the right amount of information.
\end{itemize}

\item In the argument for $\ac$ circuits, one would show a multi-switching lemma on depth 2 $\ac$ circuits. In other words, one would argue that a collection of $\ac_2$ circuits simultaneously simplify after a one random restriction is applied to all of them. Rather than the natural idea of proving a switching lemma for the analogous $\gt_2^0(k)$ circuits, we consider the hybrid class of $\g(k)\circ \{\AND,\OR\}$ circuits. It turns out a switching lemma on these simpler circuit classes suffice to depth reduce and prove bounds on $\gt_d^0(k)$ as we will see below.
\end{itemize}

\subsection{The Depth Reduction Lemma}

The multi-switching lemma gives a simplification lemma for depth 2 circuits. To extend this to constant depth circuits, we would like to iteratively decrease the depth of the circuit and induct. The argument for $\ac$ circuits was quite simple. Say we have a depth 3 $\OR \circ \AND \circ \OR$ circuit $F$. Using the switching lemma, we can say with high probability, $F|_{\rho}$ is an $\OR\circ \dt_t$ circuit. We now expand each bottom layer decision tree into an $\OR\circ \AND_t$ circuit by enumerating over all 1-paths. Consequently this simplifies $F|_{\rho}$ to a $\OR\circ(\OR\circ \AND_t) = \OR\circ \AND_t$ circuit, since an $\OR$ of $\OR$ of variables is simply a single $\OR$ over all variables involved, getting us a depth reduction from depth 3 to 2.\

What happens when we try the same argument for a $k$-$\OR\circ k$-$\AND\circ \OR$ circuit $F$? By our switching lemma, $F|_\rho$, with high probability will simplify to a  $k$-$\OR\circ \dt_t$ circuit. We can then unravel the decision trees into $\OR_{2^t}\circ \AND_t$ CNFs, resulting in a $k$-$\OR\circ \OR_{2^t}\circ \AND_t$ circuit. Here, we reach an issue: a $k$-$\OR\circ \OR$ circuit is not necessarily itself a $k$-$\OR$ function! We could have up to $(k-1)2^t$ input bits of a $k$-$\OR\circ\OR_{2^t}$ be 1 while still evaluating to $0$ (set all $2^t$ bits of $k-1$ of the bottom depth $\OR$s to be 1). The best we can do is say the function is in $\g((k-1)2^t+1)$, which is too large of a blowup in the ``$k$'' parameter for our switching lemma to handle.\

We rewind a bit to our $k$-$\OR\circ \dt_t$ circuit and unravel the bottom-layer trees to $\OR_{2^t}\circ \AND_t$ DNFs by enumerating over 1-paths. But we now make the key observation about each DNF which follows from the fact any assignment uniquely defines a path on a decision tree: \emph{any assignment of $x$ makes at most $1$ $\AND_t$ clause true}. To use more standard terminology, the DNF created from decision trees is \emph{unambiguous}. This means the pathological case above of all clauses under $k-1$ $\OR_{2^t}$ gates being satisfied cannot happen. In fact, we can prove something stronger. Since at most one clause under each $\OR_{2^t}$ gate can be satisfied in the unraveled $k$-$\OR\circ \OR_{2^t}\circ \AND_t$ circuit, the number of middle layer $\OR_{2^t}$ clauses that are satisfied will be precisely the total number of bottom layer $\AND_t$ clauses that are satisfied. Hence, a $k$-$\OR$ over the $\OR$ gates is exactly the same as a $k$-$\OR$ over the $\AND_t$ clauses themselves, and we can indeed collapse to a $k$-$\OR\circ \AND_t$ circuit! This gets us our depth reduction. A slightly more involved argument is carried out to show the more general $\g(k)\circ \dt_t$ circuit can be calculated by a $\g(k)\circ \AND_t$ circuit, but the heart of the argument is captured in the $k$-$\OR$ case itself.

\subsection{Putting It All Together}

We now have all the ingredients to simplify $\gt_d^0(k)$ circuits. The argument will be the following inductive process, where we are effectively inducting on circuits of the form $\gt_d^0(k)\circ \{\AND,\OR\}_w$ rather than $\gt_d^0$ directly. Given a $\gt_d^0(k)$ circuit, \begin{enumerate}
\item Add a trivial $(d+1)$-st layer at the bottom that is simply the identity gate (think of it as an $\AND_1$ gate)
\item By the multi-switching lemma, we know the depth 2 $\g(k)\circ \{\AND,\OR\}$ subcircuits simplify to $\dt_t$ trees with high probability, resulting in a $\gt_{d-1}^0(k)\circ \dt_t$ circuit.
\item By the depth reduction lemma, each of the bottom depth 2 $\g(k)\circ \dt_t$ subcircuits can be calculated by a $\g(k)\circ \{\AND_t,\OR_t\}$ circuit, resulting in a $\gt_{d-1}^0(k)\circ \{\AND,\OR\}$ circuit.
\item The depth has reduced by 1, so we go back to Step 2 and induct.
\end{enumerate}

This argument allows us to use the multi-switching lemma along with the depth reduction lemma to establish size lower bounds for $\gt_d^0(k)$ bounds. We show a formal argument of this outline in \Cref{sec:switch}.

\section{Preliminaries}

\subsection{Notation}
$[n] = \{1,2,\dots, n\}$ denotes the set of the first $n$ positive integers. $\binom{[n]}k$ denotes the set of all size $k$ subsets of $[n]$. $\log$ is assumed to be in base 2. This paper concerns \emph{constant}-depth circuits, and so the depth variable, $d$, should be treated as a constant. In particular hidden constants in $O(\cdot)$ or $\Omega(\cdot)$ may depend on $d$. For $S\subset[n]$, we denote $x^S = \prod_{i\in S}x_i$.

\subsection{Random Restrictions and Partial Assignments}
A \emph{partial assignment} or \emph{restriction} is a string $\rho\in \{0,1,\star\}^n$. Intuitively, a $\star$ represents an index that is still ``alive'' and hasn't been fixed to a value yet.

\  An alternative way of defining a restriction is by the set of alive variables and a ``ground assignment'' string. Given a ``$\star$ set'' $\Lambda$ and a ground assignment $z\in\bits^n$, we define $\rho(\Lambda, z)$ to be the partial assignment where we assign \[\rho(\Lambda, z)_i = \begin{cases}\star & i\in \Lambda \\ z_i & i\notin \Lambda\end{cases}\] Sometimes, $\Lambda$ may be in the form of an indicator $\bits^n$ string, where the set is defined to be the set of indices containing a 1.\ 

We also define a composition operation on partial assignments. For two restrictions $\rho^1,\rho^2$, define $\rho^1\circ \rho^2$ so that \[(\rho^1\circ\rho^2)_i =\begin{cases} \rho^1_i & \rho^1_i\neq \star \\ \rho^2_i & \rho_i^1 = \star.\end{cases} \] Intuitively, one can see this as fixing bits determined by $\rho^1$ first, and then out of the remaining alive positions, fix them according to $\rho^2$.\

A \emph{random restriction} is simply a distribution over restrictions. A common random restriction we will use is $R_p$, the distribution where each index will be assigned $\star$ with probability $p$, and $0,1$ each with probability $\frac{1-p}2$.

The main reason for defining restrictions is to observe their action on functions. Given a restriction $\rho$ and function $f:\bits^n\to \bits$, we define $f|_{\rho}: \bits^n\to\bits$ to be the function mapping $f|_{\rho}(x) := f(\rho\circ x)$. 

\subsection{Models of Computation}

\subsubsection*{Circuits}
 We measure the size of a circuit by the total number of wires (including input wires) in it. We define the width of a DNF or CNF to be the maximum number of variables in any of its clauses. We also use $k$-DNF (resp. $k$-CNF) to denote DNF (resp. CNF) of width at most $k$.
$\ac_d$ are depth $d$ circuits with unbounded fan-in whose gate set is $\{\AND, \OR, \NOT\}$. In general, if we have a gate $G$, a subscript $G_k$ will refer to its fan-in (in this case, $G$ is fixed to have fan-in $k$).
 We now define more general circuit classes that we analyze in this work.

\begin{defn}[$k$-$\OR/k$-$\AND/\ac_d(k)$] 
Define $k$-$\OR_m: \{0,1\}^m\to \{0,1\}$ to be a function that evaluates to $1$ iff $x$ contains $\ge k$ ones. Analogously define $k$-$\AND_m$ to be $0$ iff $x$ contains $\ge k$ zeros. Define $\ac_d(k)$ to be the class of depth $d$ circuits with unbounded fan-in whose gate set is $\{k'$-$\AND, k'$-$\OR, \NOT\}$ for all $k'\le k$.
\end{defn}

In more generality, we define $\g(k)$ gates and $\gt_d^0(k)$ circuits.

\begin{defn}[$\g{(k)}/\gt^0{(k)}$]
Define a gate set $\g(k)$ to be the set of all arbitrary fan-in gates such that they are constant on inputs with $\ge k$ ones (we call such gates \emph{orlike}) or are constant on inputs with $\ge k$ zeros (we call such gates \emph{andlike}). $\gt^0(k)$ is the class of constant depth circuits made by $\g(k)$ gates. 
\end{defn}

In the rest of the paper, we may write circuit classes $\gt_d^0(k)\circ\{\AND,\OR\}$ or $\g(k)\circ\{\AND,\OR\}$. In the literature, this usually refers to the circuit class whose gates above the bottom layer are in $\g(k)$, and whose bottom layer gates can either be $\AND$ or $\OR$ with no restriction on the choice. However, in this paper, assume this notation implicitly restricts $\AND$ gates to only be under orlike $\g(k)$ gates and $\OR$ gates to only be under andlike $\g(k)$ gates. \ 

On top of being an alternate generalization of $\AND/\OR$ gates, $\g(k)$ gates capture arbitrary $\ltf$s that are ``unbalanced'' in some sense. We will use the $\{\pm 1\}$ bits to define these, but one can convert between $\{0,1\}$ and $\{\pm 1\}$ via the map $b\to (-1)^b$.

\begin{defn}[Balance of an $\ltf$/$\tc^0(k)$]
Consider an arbitrary $\ltf$ $f:\{\pm 1\}\to \{\pm 1\}$ with $f(x) = \text{sgn}(\sum w_i x_i - \theta)$. Let $\sigma: [n]\to [n]$ be a permutation ordering $(w_i)$ such that $|w_{\sigma(1)}|\le \dots\le |w_{\sigma(n)}|$. Define the \emph{balance} of $f$ (denoted as $\text{bal}(f)$) to be the smallest integer $k$ such that $-\sum_{i\le k} |w_i| + \sum_{i > k}|w_i| < |\theta|$. Now denote $\tc^0(k)$ to be the class of constant depth circuits made out of $\thr$ gates with balance $\le k$. 

\end{defn}

We prove that up to negations in the inputs and output, $\thr$ gates with balance $k$ are in $\g(k)$ in the appendix (\Cref{thm:thrgk}). All results in this paper hold for $\tc^0(k)$, but from now on, we will only refer to $\gt^0(k)$ as it is the more general class.

\subsubsection*{Decision Trees}
We assume knowledge of decision trees (see Definition 3.13 in \cite{odonnel} for a reference). We will be using slightly more complex models of decision trees in this work.

\begin{defn}[Partial Decision Trees]
    For a collection of functions $\calF = \{F_1,\dots, F_m\}$, we say $\calF$ can be computed by an $r$-partial depth-$t$ $\dt$ if there exists a singe depth $r$ tree such that for all $F_i$ and paths $\pi$ of $T$, $F_i|_\pi$ can be computed by a depth $t$ decision tree (here, $F|_\pi$ is $F$ acted on by the restriction induced by taking path $\pi$ down $T$). 
\end{defn}

\begin{defn}[$(d,\calC)$-tree]
    Let $d$ be an integer and $\calC$ a computational model (e.g. a circuit class). A function is computable by a $(d,\calC)$-tree if it is computable by a depth $t$ decision tree with $\calC$ functions as its leaves. That is, there exists a depth $d$ decision tree $T$ such that for every path $\pi$ in $T$, $F|_\pi\in \calC$.
\end{defn}

\subsection{Pseudorandomness and Probability}

We will use various pseudorandom primitives and terminology. We will use $U_n$ to denote the uniform distribution over $n$ bits unless specified otherwise.

\begin{defn}[$\eps$-error PRG/Seed Length]

A distribution $D$ over $\bits^n$ is called an $\eps$-error PRG for a computational model $\calC$ if for all $C\in\calC$, \[|\E_{x\sim U_n}[C(x)]-\E_{x\sim D}[C(x)]|\le \eps\] The seed length $s$ of $D$ is defined to be the minimal quantity $s$ such that the following is true: there exists a polytime computable function $G: \bits^s\to\bits^n$ such that the distribution of $G(z)$ over $z\sim U_s$ is exactly $D$.
    
\end{defn}

\begin{defn}[$(\eps, k)$-wise independent source]

A distribution $D$ over $\bits^n$ is an $(\eps,k)$-wise independent source if for all $1\le i_1 < \dots < i_k\le n$ and $\alpha\in\bits^k$, \[|\Pr_{x\sim D}[x_{i_1}x_{i_2}\dots x_{i_k} = \alpha] - 2^{-k}| < \eps.\]  
    
\end{defn}

There exists constructions of these sources with seed length $O(\log\log n + k + \log(1/\eps))$ \cite{aghp}.

\begin{defn}[$k$-wise Independent Hash Family]
Let $\calH$ be a distribution over hash functions mapping $\bits^n\to\bits^m$. We say that $\calH$ is $k$-wise independent if for any $k$ input-output pairs $(x_1,y_1),\dots, (x_k, y_k) \in \bits^{n}\times \bits^m$ where $x_1,\dots, x_t$ are distinct, 
it holds that
\[
\Pr_{h\sim \calH} [ \forall i\in [k], h(x_i) = y_i ]  = 2^{-km}.
\]

\end{defn}

Such functions can be sampled using $O(k(n+m))$ bits (Chapter~3.5.5 of \cite{vad12}).

\begin{defn}[$k$-wise $p$-bounded Subset]
Let $\Lambda$ be a random subset of $[n]$. $\Lambda$ is a $k$-wise $p$-bounded subset iff for all subsets $S\subset[n]$ of size $\le k$, $\Pr_{\Lambda}[S\subset \Lambda]\le p^{|S|}$. 
\end{defn}
For example, $R_p$ is $n$-wise $p$-bounded.  

\subsection{Fourier Analysis}

Every Boolean function $f: \{\pm 1\}^n \to\{\pm 1\}$ has a unique representation as a multilinear real polynomial \[f(x) = \sum_{S\subset [n]} c_Sx^S.\]
Given $f$, we can think of the Fourier transform of $f$, $\widehat{f}$ to be a function mapping $2^{[n]}\to \R$ such that $\widehat{f}(S) = c_S$. This is well defined by the uniqueness of the polynomial representation of $S$. One can explicitly compute $\hat{f}(S) = \E_x[f(x)x^S]$. By Parseval's, one can derive $\sum_{S\subset[n]}\widehat{f}(S)^2 = 1$. There are various quantities involving the Fourier coefficients that we will work with.

\begin{defn}[Fourier Tails]
For a Boolean function $f$, define \[W^{\ge k}[f] := \sum_{|S|\ge k}\widehat{f}(S)^2.\]
    
\end{defn}

\begin{defn}[Discrete Derivative/Influence]

For Boolean $f$ and $i\in [n]$, define the discrete derivative \[D_if(x) = \frac{f(x^{(i\to 1)})-f(x^{(i\to -1})}2\] where $x^{(i\to b)} = (x_1,\dots, x_{i-1},b,x_{i+1},\dots, x_n)$. Now for $S\subset[n]$ with $S = \{i_1,\dots, i_k\}$, define \[D_Sf = D_{i_1}D_{i_2}\dots D_{i_k}.\] Now for $S\subset[n]$, define the influence \[\text{Inf}_S(f) = \E_{x\sim \{\pm 1\}^n}[D_Sf(x)^2].\] Finally, define the degree $k$ influence \[\text{Inf}^k(f) = \sum_{|S| = k}\text{Inf}_S(f).\]
    
\end{defn}

\section{Simplification Theorem of $\gt_d^0(k)$ Circuits}
\label{sec:switch}
\begin{thm}
\label{thm:switch}
Let $F$ be computable by a depth-2 $\g(k)\circ\{\AND,\OR\}_w$ circuit. Let $\Lambda$ be a $(t+w)$-wise $p$-bounded subset of $[n]$, and $x$ a uniform string. Then  \[\Pr_{\Lambda, x}[\dt(F|_{\rho(\Lambda, x)})\ge t]\le (20pw)^t2^k. \]

\end{thm}

\begin{proof}
The proof will follow that of Section 5 in \cite{Lyu}. We urge the reader to first read the overview given in \Cref{sub:switch}. As discussed there, the main differences between this proof and the one presented there are present in the constructions of the canonical decision tree and witness searchers, the definition of witnesses, and the counting of partial witnesses. These are altered to support the more general $\g(k)$ gates. Besides this, the general proof strategy remains the same. Let $m$ be the fan-in of the $F$. We present a procedure that constructs a decision tree (which we deem the ``Canonical Decision Tree'').
\RestyleAlgo{ruled}

\SetKwComment{Comment}{/* }{ */}
\SetKwBlock{init}{initialize:}{}
\SetKwInput{kwInput}{Input}
\SetKwInput{kwReturn}{return}
\begin{algorithm}[hbt!]
\caption{Canonical Decision Tree}\label{alg:one}
\kwInput{$(\text{orlike }G(k))\circ \AND_w$ circuit $F = G(C_1,\dots, C_m)$, black-box access to a string $\alpha\in\{0,1\}^n$.}
\init{$j^*\gets 0$ \\ $x\gets (\star)^n$ \\ $ctr\gets 0$}
\While{$j^* < m$}{
    Find the first $j>j^*$ such that $C_{j}(x)\not\equiv 0$. If no such $j$ exists, exit the loop. \\ 
    $B_{j}\gets$ the set of unknown variables in $C_{j}(x)$ (may be empty).\\
    Query $\alpha_{B_{j}}$.\\
    Set $x_{B_{j}}\gets \alpha_{B_{j}}$.\\
  \If{$C_{j}(x) = 1$}{
    $ctr\gets ctr + 1$\;
    \If{$\ctr = k$}{
        \textbf{return} $G(1^m)$} 
  }
  $j^*\gets j$\\
  }

  \textbf{return} $F(x\circ 0^n)$.

\end{algorithm}

The difference between the CDT defined in \cite{Lyu} and the one presented here is the use of $ctr$. Intuitively, this is added in to keep track of the number of satisfied clauses we see before we reach our limit of $k$. We rigorously prove this in the following claim.

\begin{claim}
The CDT correctly outputs $F(\alpha)$.
\end{claim}

\begin{proof}
The CDT scans the clauses in order to find the first one not fixed to zero. There are only two return statements in the algorithm, so we consider the two cases of terminating on each one. Suppose we terminate at the first return statement and output $F(1^m)$. Notice $ctr$ is incremented each time the CDT encounters a satisfied clause. Therefore, when $ctr=k$, at least $k$ $C_i$ evaluate to 1, and therefore $F(C_1,\dots, C_m) = F(1^m)$ by virtue of $F\in\g(k)$ being orlike, proving correctness. Now suppose we terminate after the while loop and output $C(x\circ 0)$. If CDT finishes the while loop without terminating, that must mean all clauses must be determined by the partial assignment $x$. This is because for any clause $C_{j'}$, if the clause wasn't already determined in the algorithm when $j^* = j'$, all unknowns of $C_{j'}$ would have been queried and fixed in the partial assignment $x$, thereby determining it. Therefore, in this case, $C(x)$ is determined, and in particular is equal to $C(x\circ 0^n)$. 

\end{proof}

Therefore, CDT is indeed a decision tree computing $F$. Define $\func{CDT}(F)$ to be the depth of the canonical decision tree $T_{F}$.If $\rho$ is bad (i.e. $\dt(F|_\rho)\ge t)$, then clearly $\func{CDT}(F|_\rho)\ge \dt(F|_\rho)\ge t$ and so $T_{F|_\rho}$ will result in at least $t$ queries for $\emph{some}$ choice of $\alpha$ (this is equivalent to saying that $\emph{some}$ path of $T_{F|_{\rho}}$ must have length $\ge t$). We define a witness that will effectively be the transcript of the algorithm on this particular $\alpha$.

\begin{defn}
\label{defn:wit}
Let $F$ be the circuit described above and $\rho$ a restriction. Let $t\ge 1$. Consider the tuple $(r,\ell_i,s_i,B_i,\alpha_i)$ where 
\begin{itemize}
\item $r\in [1,t+k]$ is an integer
\item $(\ell_1,\dots, \ell_r)\in [m]^r$ is an increasing list of indices
\item $(s_1,\dots, s_r)$ is a list of non-negative integers, at most $k$ of which are allowed to be $0$, such that $s:= \sum_{i=1}^r s_i\in [t,t+w-1]$
\item $(B_1,\dots, B_r)$ is a list of (potentially empty) subsets of $[w]$ satisfying $|B_i| = s_i$.
\item $(\alpha_1,\dots, \alpha_r)$ is a list of (potentially empty) bit strings satisfying $|\alpha_i| = s_i$.
\end{itemize}

$(r,\ell_i,s_i,B_i,\alpha_i)$ is called a $t$-witness for $\rho$ if there exists an $\alpha\in\{0,1\}^n$ such that 

\begin{itemize}
    \item When we run $T_{F|_\rho}$ on $\alpha$, $C_{\ell_i}$ is the $i$-th term queried by $T_{F|_\rho}$.
    \item $T_{F|_\rho}$ queries $s_i$ variables in $C_{\ell_i}$, and the relative location of those variables within $C_{\ell_i}$ are specified by set $B_i$.
    \item $T_{F|_\rho}$ receives $\alpha_i$ in response to its $i$-th batch query.
\end{itemize}

The \emph{size} of the witness $(r,\ell_i,s_i,B_i,\alpha_i)$ is defined to be $s:=\sum_{i=1}^r s_i$. We may denote the size of a witness $W$ as $\func{size}(W)$.
\end{defn}

\begin{claim}
\label{claim:wit}
    For every $\rho$ such that $\dt(F|_{\rho})\ge t$, there exists a $t$-witness for $\rho$
\end{claim}

\begin{proof}
We simply run $T_{F|_{\rho}}$ on $\alpha$ that causes at least $t$ queries to be issued. We then record the transcript \emph{until the number of variables queried exceeds} $t$, after which we halt. To be more explicit, we let $r$ be the number of times the CDT stops at a clause before either outputting a bit or exeeding $t$ queries. At the $i$th stop, say on clause $C_j$, we set $\ell_i = j$, $s_i$ to be the number of unknown variables queried (which may be 0), $B_i$ to be the subset of $[w]$ indicating the relative positions of the variables in the clause, and $\alpha_i$ being the query replies received from the black-boxed $\alpha$. We can first verify this creates a valid tuple.

\begin{itemize}
    \item $r\in [1,t+k]$. Every time we stop at a clause, either it evaluates to 1 and we increment $ctr$, or we have to query at least 1 variable. Since $ctr$ can incremented at most $k$ times and we can query variables from at most $t$ clauses before reaching our quota of $t$ queried variables, it follows we stop at most $t+k$ times.
    \item $(\ell_1,\dots, \ell_r)\in [m]^r$ is an increasing list of indices since the CDT linearly sweeps the clauses in increasing index order.
    \item $(s_1,\dots, s_r)$ is a list of non-negative integers, at most $c$ of which are allowed to be $0$, such that $s:= \sum_{i=1}^r s_i\in [t,t+w-1]$. A particular $s_i$ being zero implies that the CDT's $i$th stop was at a clause $C_j$ that was already determined to be 1, and hence $ctr$ was incremented. Since $ctr$ can be incremented at most $k$ times, at most $k$ of the $s_i$'s are zero. $s$ is the total number of variables queried before halting. Notice after the penultimate clause is queried, there are $ < t$ clauses queried. Consequently when the ultimate clause is queried, there clearly will be $<t+w$ variables queried, since $k$ is the width of the clause. Of course, since the transcript halted after this clause, $\ge t$ variables had to be queried.
    \item $(B_1,\dots, B_r)$ is a list of (potentially empty) subsets of $[w]$ satisfying $|B_i| = s_i$ trivially by construction.
    \item $(\alpha_1,\dots, \alpha_r)$ is a list of (potentially empty) bit strings satisfying $|\alpha_i| = s_i$ trivially by construction.
\end{itemize}

We can then easily see by construction of the tuple, it is indeed a $t$-witness for $\rho$.
\end{proof}
We note that the difference between Definition 4 in \cite{Lyu} and the one here is the relaxation to allow $(s_1,\dots, s_r)$ to contain up to $k$ zeros, rather than to all be positive. As evident in the proof, this is to handle cases the CDT encounters a clause that was already fixed to 1, which causes the corresponding $s_i$ value to be $0$. This wasn't recorded in Lyu's witness definition because in the case of CNFs, one satisfied clause determines the value of the circuit, and the CDT immediately halts. Why do we still record that the CDT didn't query any variables at a clause instead of just ignoring this behavior and moving on? It turns out if we don't include this piece of information, the witness searcher we create will not have enough information to reconstruct the whole witness (see the ``balancing act'' discussion in \Cref{subsub:comparison}).

We now move on to define partial witnesses.

\begin{defn}
\label{defn:partwit}
Let $F$ be a circuit and $\rho$ a restriction. We call $(r,s_i,B_i,\alpha_i)$ a partial $t$-witness for $\rho$ if there exists $(\ell_1,\dots,\ell_r)$ such that $(r,\ell_i,s_i,B_i,\alpha_i)$ is a $t$-witness for $\rho$.
\end{defn}

We note the following important claim.

\begin{claim}
\label{claim:unique}
If $P$ is a partial witness for $\rho$, then there exists exactly one list of integers $(\ell_i)$ such that $(\ell_i,P)$ is a witness for $\rho$.
\end{claim} 
\begin{proof} 
By construction of \Cref{alg:one}, $\ell_1$ must be be the index of the first clause not fixed to $0$ by $\rho$.  But now, we notice $\ell_2$ must be the index of the first clause after $C_{\ell_1}$ not fixed to $0$ by $\rho\circ \alpha_1$. We then continue this induction to get our unique list $(\ell_i)$, Where $\ell_j$ will be forced to be the index of the first clause after $C_{\ell_{j-1}}$ that is not fixed to $0$. 
\end{proof} 

Therefore, by Claims \ref{claim:wit} and\ref{claim:unique}
\begin{align}\Pr_\rho[\dt(F|_\rho)\ge t]\le \sum_{(\ell_i,P)} \Pr_\rho[(\ell_i,P)\text{ is a $t$-witness for }\rho]\le \sum_P \Pr_\rho[P\text{ is a partial $t$-witness for }\rho]]\label{eqn:4}\end{align} where $P$ ranges over all partial $t$-witness tuples.

Going back to our proof, we now define our \emph{witness searcher} ${\cal S}$.
\clearpage

\SetKwComment{Comment}{/* }{ */}
\SetKwBlock{init}{initialize:}{}
\SetKwInput{kwInput}{Input}
\SetKwInput{kwReturn}{return}
\begin{algorithm}[hbt!]
\caption{Witness Searcher ${\cal S}$}\label{alg:two}
\kwInput{(orlike $\g(k))\circ \AND_w$ circuit $F(C_1,\dots, C_m)$, ground assignment $z\in\{0,1\}^n$, partial witness $W = (r,s_i,B_i,\alpha_i)$ .}
\init{$j^*\gets 0$ \\ $x\gets (\star)^n$ \\ $ctr\gets 1$}
\While{$ctr\le r$}{
    \While{$j^* < m$}{
    Find the first $j>j^*$ such that $C_{j}(z)\equiv  1$. If no such $j$ exists, exit the inner while loop. \\ 
    $\ell_{ctr}\gets j$\\
    Query $\alpha_{B_{j}}$.\\
    Set the $B_{ctr}$ portion of $z$ to be $\alpha_{ctr}$.\\
    $ctr\gets ctr + 1$\\
  $j^*\gets j$\\
  }
  }
  
  \textbf{return} $(\ell_i, W)$

\end{algorithm}

We now prove the following essential property about ${\cal S}$, stating in a probabilistic way, it can use a partial witness to reconstruct a total witness for $\rho$.

\begin{lem} 
\label{lem:search}
Let $P$ be a partial witness, and let $s$ be its size. Define a restriction $\rho$ to be \emph{good} for $P$ if $P$ is a partial $t$-witness for it.
\[\Pr_{z}[{\cal S}(z, P)\text{ is a $t$-witness for $\rho(\Lambda,z)$} | \text{$\rho(\Lambda, z)$ is good for }P] = 2^{-s}\]

Furthermore this event is solely dependent on $z_I$, where $I$ is the set of variable indices referred to by the unique completion of $P$ with respect to $\rho$.
\end{lem}

\begin{proof}
If $\rho = \rho(\Lambda,z)$ is good, then by \Cref{claim:unique} we know there exists unique $\ell_i$ such that $(\ell_i, P)$ witnesses $\rho$. In particular, we know that $\Lambda$ must contain all the indices $I$ that $(\ell_i, B_i)$ identify. Let $I_j$ be the index set identified by $\ell_j$ and $B_j$ (so $I = I_1\sqcup \dots \sqcup I_r$). Now condition on a fixed $\rho$. This means all bits in $z$ not covered by $\Lambda$ are fixed. In particular, only source of randomness left are the bits covered by $\Lambda$, which is a superstring of $z_I$.  We now claim every $z_{I_j}$ is assigned the unique bit string such that $C_{\ell_i}|_{z_{I_j}}\not\equiv 0$ (not forced to be unsatisfied) iff ${\cal S}$ successfully outputs $(\ell_i, P)$. This consequently proves the lemma, since this has probability $2^{-|I|} = 2^{-s}$ of happening.\ 

By construction of $T_{F|_\rho}$ we know that all clauses before $C_{\ell_1}$ was falsified by $\rho$. Upon inspection, we see ${\cal S}$ correctly skips past these clauses (as $z$ is a completion of $\rho$). Now we note $C_{\ell_1}$ was not fixed to $0$ $\rho$, causing $T_{F|_\rho}$ to query all unknowns in $C_{\ell_1}$ at the time (which might be nothing if $C_{\ell_1}$ was fixed to 1), which is $x_{I_1}$. Inspecting ${\cal S}$, we see ${\cal S}$ will set the correct $\ell_1$ iff $C_{\ell_1}(z)$ is satisfied iff $z_{I_1}$ is assigned the unique string such that $C_{\ell_i}|_{z_{I_j}}\not\equiv 0$ (since all variables outside $I_1$ ocurring in $C_{\ell_1}$ is fixed by $\rho$). ${\cal S}$ then (importantly) replace $z_{I_1}$ with $\alpha_{I_1}$ so that all variables encountered thus far are assigned exactly as $T_{F|_\rho}$ did. \ 

We then repeat this argument $r$ times, noting that due to $z\circ \alpha_{I_1}\circ\dots\circ\alpha_{I_j}$ being a completion of $\rho\circ \alpha_{I_1}\circ\dots\circ \alpha_{I_j}$, ${\cal S}$ rightfully skips all clauses between $C_{\ell_j}$ and $C_{\ell_{j+1}}$. We then similarly argue that ${\cal S}$ will set $\ell_{j+1}$ to be the $j+1$st clause $T_{F|_\rho}$ queries iff $z_{I_{j+1}}$ is the unique string such that $C_{\ell_{j+1}}|_{z_{I_{j+1}}}\not\equiv 0$.

\end{proof}

Combining \Cref{eqn:4} and \Cref{lem:search}, it follows that \begin{align}\Pr_\rho[\dt(F|_\rho)\ge t] & \le \sum_P \Pr_\rho[P\text{ is a partial $t$-witness for }\rho] \nonumber\\ & \le \sum_P \E_{\Lambda}\frac{\Pr_z[{\cal S}(z,P)\text{ is a $t$-witness for }\rho(\Lambda, z)]}{\Pr_z[{\cal S}(z,P)\text{ is a $t$-witness for }\rho(\Lambda, z)| \rho(\Lambda, z)\text{ is good}]} \nonumber \\ & \le  \sum_P 2^{\func{size}(P)}\E_z\Pr_\Lambda[{\cal S}(z,P)\text{ is a $t$-witness for }\rho(\Lambda, z)] \label{eqn:nump}\end{align} Notice a necessary condition for a restriction $\rho(\Lambda, z)$ to be $t$-witnessed by a size-$s$ $W$ is for $\Lambda$ to cover the $s$ variables that $W$ recorded as the CDT needing to query, which happens with probability $\le p^s$ (as $s\le t+w$ and $\Lambda$ is $(t+w)$-wise $p$-bounded). Hence, every term in the sum in (\ref{eqn:nump}) can be bounded by $p^{\func{size}(P)}$ and it remains to find the number of partial $t$-witness tuples $P$. \ 

For a fixed $s$, we can bound the number of potential partial witnesses naively by noting 
\begin{itemize}
    \item the number of choices of $(r,s_i)$ can be bounded by the number of ways to write $s$ as the sum of at most $s+k$ nonnegative integers, which is $\sum_{r=1}^{s+k}\binom{s+r}r\le\sum_{r=1}^{s+k}\binom{2s+k}r \le 2^{2s+k}$ (notice that we get a larger count here than the analogous quantity of $2^{2s}$ in Lyu's proof \cite{Lyu}, which is a side effect of looking at a more complicated circuit class),
    \item the choices for $(B_i)$ can be bounded by $\prod_i \binom{w}{s_i}\le w^s$,
    \item and the choices for $(\alpha_i)$ can be bounded by $2^s$,
\end{itemize}   giving a total count of $(8w)^s2^k$. Combining this count with the previous paragraph's observation and (\ref{eqn:nump}), while remembering to sum over all sizes, we derive \[\Pr_{\rho}[\dt(F|_\rho)\ge t]\le \sum_{s=t}^{t+w-1}(2p)^s(8w)^s2^k =\sum_{s=t}^{t+w-1}(16pw)^s2^k\le (20pw)^s2^k.\]
\end{proof}

\begin{remark} One may ask whether the failure probability of $(20pw)^t2^k$ tight. We show that $\pargate_{kw}$ can be expressed as a $\g(k)\circ \AND_w$ cirucit and prove this saturates the above bound in the Appendix (\Cref{apx:tight}).
\end{remark}

After defining witnesses, partial witnesses, and witness searchers for $\g(k)\circ \{\AND,\OR\}$ circuits, we notice that Lyu's proof of the multi-switching lemma directly goes through with these definitions with zero changes (even down to the exact algorithm of the canonical partial decision tree and global witness searcher). Due to this, we defer the proof of the multi-switching lemma to the appendix. However, we do highlight here what properties about the circuit class is needed in order to invoke Lyu's lift from a switching lemma to a multi-switching lemma. The key properties needed were that 

\begin{itemize} \item the number of partial witnesses for a depth $t$ canonical decision tree needed to be small, and 
\item there needed to exist a witness searcher function ${\cal S}$ such that for all $\rho$ and a partial witness for $\rho$, ${\cal S}$ recovers the full witness with decent probability over an advice string.
\item given a complete witness, there needed to be a small chance that a random restriction witnessed it\end{itemize}

\begin{thm}
\label{lem:multiswitch}
Let ${\cal F} = \{F_1,\dots, F_m\}$ be a list of 
$\g(k)\circ \AND_w$ circuits on $\{0,1\}^m$. Then \[ \Pr_{\rho\sim R_p}[{\cal F}|_\rho\text{ does not have $r$-partial depth-$t$ $\dt$}]\le 4(64(2^km)^{1/r}pw)^t \]
\end{thm}

\begin{proof}
    See \Cref{apx:multiswitch}.
\end{proof}

\begin{remark}
At this point, we can observe an aspect of the expression that illuminates an important unifying flavor of the rest of the results. Notice that the only difference in the failure probability expression between the standard $\ac$ multi-switching lemma in (lyu) and the above one for $\g(k)$ is that every occurrence of $m$ is multiplied by a factor of $2^k$. This means if constants in the exponent can be ignored, the multi-switching lemma asymptotically gives the same result as the $\ac$ version if $2^k \approx m$. In particular, we can expect any result for size $s$ $\ac$ circuits to immediately extend to $\gt^0(\log s)$ circuits with no loss in parameters! This will be demonstrated in various settings in future sections.
\end{remark}

With our multi-switching lemma in hand, we can simplify depth 2 circuits with high probability. To extend this to constant depth circuits, we also require a depth reduction lemma. In the case of $\ac$, this was trivial enough to embed in the main proof, but in the case of $\gt^0(k)$, we need to be more delicate and use more specific properties of decision trees.

\begin{lem}
\label{lem:depthred}
 Any depth 2 circuit of the form $\g(k)\circ \dt_w$ with top gate fan-in $m$ can be expressed as a circuit in $G(k)\circ \{\AND,\OR\}_w$ of size $m2^{w}$.
\end{lem}

\begin{proof}
Say the circuit we start with is $F(D_1,\dots, D_m)$, where $D_i$ are the bottom layer depth $w$ decision trees. Assume $F$ is orlike (the andlike case is analogous). By enumerating over all 1-paths, expand out each $D_i$ as an $\OR$ of $\AND$s, namely $C^i_1\vee C^i_2\vee\dots\vee C^{i}_{2^w}$. Now define a function $F'$ over $m2^w$ bits, where \[F'(x_1^1,\dots x_{2^w}^1,x_1^2,\dots, x_{2^w}^m) = \begin{cases} F(1^m) & \sum_{i,j}x_j^i\ge k \\ F(\bigvee_{i=1}^{2^w} x_j^1,\dots, \bigvee_{i=1}^{2^w} x_j^m) & \text{otherwise}   \end{cases}\] Clearly by construction, $F'\in \g(k)$. Therefore to prove the lemma, it suffices to show that over all input assignments, $F(D_1,\dots, D_m) = F'(C_1^1,\dots C_{2^w}^1,C_1^2,\dots, C_{2^w}^m)$.   \ 

If $\ge k$ of the $D_i$ are satisfied, we know since $F$ is an orlike $\g(k)$ function, $F(D_1,\dots, D_m) = F(1^m)$. This also clearly implies $\ge k$ of the $C^i_j$ are satisfied. Therefore by construction of $F'$, $F'(C_1^1,\dots, C_{2^w}^m)$ also evaluates to $F(1^m)$. \ 

If $ < k$ of the $D_i$ are satisfied, then we need to use the following observation. For any assignment of inputs, at most one of the clauses $C^i_1,\dots, C^i_{2^m}$ can be satisfied for each $i$, since each assignment uniquely defines a path in a decision tree. In more conventional terms, the DNF created by the decision tree $D_i$ is \emph{unambiguous}. Therefore the amount of $D_i$ satisfied is exactly equal to the number of $C^i_j$ satisfied, and so $<k$ clauses $C_i^j$ are satisfied. This forces us into the second case of the piecewise definition of $F'$, and  so \[F'(C_1^1,\dots, C_{2^w}^m) = F(\bigvee_{i=1}^{2^w} C_j^1,\dots, \bigvee_{i=1}^{2^w} C_j^m) = F(D_1,\dots, D_m)\] as desired. 
\end{proof}

We have finally built up the tools to prove our main result: a constant depth simplification lemma.

\begin{thm}
\label{thm:constdepth}
Let $G$ be any gate, and let $F$ be a $G\circ \gt_d^0(k)$ circuit of size $m$. Then for $p = \frac{1}{128(m2^k)^{1/w}}(128w(m2^k)^{1/w})^{-d+1}$ and any $t\ge 1$, \[ \Pr_{\rho\sim R_p} [F|_\rho\text{ is not computed by a $((2^d-1)t,G\circ \dt_w)$-decision tree}]\le 4d\cdot 2^{-t}  \]
\end{thm}

\begin{proof}
WLOG assume the circuit is layered (all paths down the circuit are of length exactly $d+1$). We first append an extra layer of $\{\AND,\OR\}_1$ gates to the bottom of the circuit so that the input level fan-in is $1$. We then apply a random restriction $\rho_0\sim R_{p_0}$ with $p_0 = \frac{1}{128(m2^k)^{1/w}}$ and use \Cref{lem:multiswitch} on all the depth-2 subcircuits to deduce that \[ \Pr_{\rho_0\sim R_{p_0}}[F|_\rho\text{ is not computed by $(t,G\circ \gt^0_{d-1}(k)\circ \dt_w)$-decision tree}]\le 4\cdot 2^{-t}. \] 

Letting $F^{(0)}$ be a good tree from above which does simplify, we see that there are at most $2^t$ leaves of the partial decision tree, with each leaf containing a $G\circ \gt^0(k)_{d-1}\circ \dt_w$ circuit (which we will refer to as ``leaf-circuits''). By \Cref{lem:depthred}, these circuits can be simplified to $G\circ \gt^0_{d-1}(k)\circ \{\AND,\OR\}_w$ circuits.  We apply \Cref{lem:multiswitch} on the depth-2 subcircuits of a particular leaf-circuit with $p_1 = \frac{1}{128w(m2^k)^{1/w}}$, using $2t$ instead of $t$, union bound over all $2^t$ leaves, and then apply \Cref{lem:depthred} to get that \[\Pr_{\rho_1\sim R_{p_1}}[F^{(0)}|_{\rho_1}\text{ is not a $(t+2t,G\circ \gt_{d-2}^0(k)\circ \{\AND,\OR\}_w$-decision tree}]\le 4\cdot 2^{-2t}\cdot 2^t = 4\cdot 2^{-t}. \]

Iterating this argument $d-2$ more times, where we apply \Cref{lem:multiswitch} on the depth 2 subcircuits using $p_i = \frac{1}{128w(m2^k)^{1/w}}$ and $2^it$ instead of $t$ on the $i$th iteration, and then union bound over all $2^{(2^i-1)t}$ leaves, we get that on the $i$th iteration, our desired single depth simplification happens with probability $2\cdot 2^{-t}$. If the desired simplifications happen on all iterations, we result in a $((2^d-1)t,G\circ \dt_w)$-decision tree with probability at most $\sum_{i=0}^{d-1} 4\cdot 2^{-t} = 4d\cdot 2^{-t}$ (via a union bound over the $d$ iterations) and with a restriction from $R_p$ where $p=\prod p_i = \frac{1}{128(m2^k)^{1/w}}\cdot (128w(m2^k)^{1/w})^{-d+1}.$ The conclusion follows.
\end{proof}

With this theorem, we can let $G$ be a $\g(k)$ gate to get the following corollary.

\begin{cor}
\label{cor:gt0}
    Let $C$ be a $\gt_d^0(k)$ circuit of size $m$ and let $p =\frac{1}{40(128(k+\log m))^{d-1}}$. Then \[\Pr_{\rho\sim R_p} [\dt(C|_\rho)\ge t] \le 2\cdot 2^{-\frac{t}{2^d-1}+k}\]
\end{cor}

\begin{proof}
    Applying \Cref{thm:constdepth} to $C$ with $w = k+\log m$, it follows for $p_1 = \frac{1}{128(128(k+\log m))^{d-2}}$, \[\Pr_{\rho\sim R_{p_1}} [C|_\rho\text{ is not computed by a $((2^{d-1}-1)t,\g(k)\circ \dt_{k+\log m})$-decision tree}]\le 4d\cdot 2^{-t}.\]

    Fix a $\rho$ such that $C$ simplifies to such a tree, $T$. By \Cref{lem:depthred}, the leaf circuits simplify to $\g(k)\circ \{\AND,\OR\}_{k+\log m}$ circuits. Let $\ell$ be a leaf, and let $C_\ell$ be the associated leaf-circuit.  By \Cref{thm:switch}, we know that for $p_2 = 1/40w$, $\Pr_{\tau\sim R_{p_2}}[\dt(C_\ell|_\tau)\ge 2^{d-1}t]\le 2^{-2^{d-1}t}2^k$. Union bounding over all $\le 2^{(2^d-1)t}$ leaves $\ell$, it follows that \begin{align}\Pr_{\rho\sim R_{p_1},\tau\sim R_{p_2}} [C|_{\rho\circ \tau}\text{ is not computed by a $((2^{d-1}-1)t,\dt_{2^{d-1}t})$-} & \text{decision tree}] \nonumber\\ &  \le 2^{(2^d-1)t}\cdot2^{-2^{d-1}t}2^k + 4d\cdot 2^{-t} \nonumber\\ 
    & \le 2\cdot 2^{-t+k} .\label{eqn:dt}\end{align}

    Because $\rho\circ \tau\sim R_{p_1p_2}$, $p:=p_1p_2 = \frac{1}{40(128(k+\log m))^{d-1}}$, and a $((2^{d-1}-1)t,\dt_{2^{d-1}t})$-tree is simply a $\dt_{2^d-1}$, (\ref{eqn:dt}) implies \[\Pr_{\rho\sim R_p}[\dt(C|_\rho)\ge (2^d-1)t]\le 2\cdot 2^{-t+k}. \] The desired result then follows after a change of variables from $t\to t/(2^d-1)$.
\end{proof}

\section{Applications of The $\gt^0(k)$ Simplification Theorem}

\subsection{Exponential Lower Bounds Against Parity}
\label{sec:constdepth}
Given \Cref{thm:switch} and \Cref{cor:gt0}, we can establish correlation bounds of $\gt_d^0(k)$ circuits against $\pargate$ (parity).

\begin{thm}
\label{thm:paritycorr}
Let $C\in \gt_d^0(k)$ have size $m$ and let $\pargate$ be the parity function. Then the correlation of $C$ and $\pargate$ is  \[\E_{x\sim \{0,1\}^n}[(-1)^{C(x) + \pargate(x)}]\le 2^{-\Omega(n/(k+\log m)^{d-1})+k}.\]

\end{thm}

\begin{proof}
The uniform distribution is equivalent to performing a fair random restriction (a random restriction where non-star variables are set to a uniform bit), and then filling in the $\star$s with uniform bits. We will show that under a fair random restriction, $C$ will become constant with high probability while $\pargate$ becomes a parity over the live variables. Averaging over these live variables then gives a correlation of zero. The total correlation is then the probability $C$ doesn't become constant.  \ 

Let $p =\frac{1}{40(128(k+\log m))^{d-1}}$. Applying \Cref{cor:gt0}, we see that \[\Pr_{\rho\sim R_p} [\dt(C|_\rho)\ge pn/4] \ge 2\cdot 2^{-\frac{pn}{4(2^d-1)}+k}.\] By a Chernoff bound, we know $\rho$ will have $\ge pn/2$ stars with $\ge 1-2^{-pn/8}$ probability. Let $\calE$ be the event both of these events happen, and fix such a $\rho$. Consider performing a random walk down the depth $\le pn/4$ decision tree (start at the root and iteratively pick which of the 2 children to travel to uniformly, effectively filling in $\le pn/4$ of the variables with uniform bits), which induces a random restriction $\tau$. No matter which path restriction $\tau$ was taken, $C|_{\rho\circ\tau}$ becomes constant, while $\pargate|_{\rho\circ\tau}$ becomes a parity over $\ge pn/2-pn/4 = pn/4$ variables. The correlation of these two functions is trivially $0$. Therefore,    
\begin{align*}
    \E_{x\sim \{0,1\}^n}[(-1)^{C(x) + \pargate(x)}] & = |\E_{\rho}\E_x[(-1)^{C|_\rho(x)+\pargate|_\rho(x)}]| \\ & \le \Pr[\neg \calE] + \E_\rho[|\E_x[(-1)^{C|_\rho(x)+\pargate|_\rho(x)}]| \big| \calE] \\ & \le 2\cdot 2^{-\frac{pn}{4(2^d-1)}+k} + 2^{-pn/8} \E_\rho[\E_{\tau}|\E_x[(-1)^{C|_{\rho\circ\tau}(x)+\pargate|_{\rho\circ\tau}(x)}]| \big| \calE] \\ & \le 2^{-\Omega(n/(k+\log m)^{d-1})+k}.
\end{align*}

\end{proof}

As an application, we can observe that for $0\le k \le .1n^{1/d}$ one can set $m = 2^{\Theta((n/k)^{\frac{1}{d-1}})}$ in the above lemma such that the correlation is $< 1/2$, yielding us the following corollary.

\begin{cor}
\label{cor:lb}
For some absolute constant $C$, integer $d$ and $0\le k \le .1n^{1/d}$, $\gt_d^0(k)$ circuits computing $\pargate_n$ requires size $2^{\Omega((n/k)^{\frac{1}{d-1}})}$.
\end{cor}
This is an interesting result in multiple ways. First, notice the dependence of the lower bound on the ``$k$'' parameter is extremely tight and the corollary becomes absurdly false if $k = n^{1/d}$. This is seen by the fact $\pargate_{n^{1/d}}\in \g(n^{1/d})$, and so one can create a size $O(n^{1-1/d})$ $\gt_d^0(n^{1/d})$ formula computing $\pargate_n$ simply by having the $i$th depth from the bottom have $n^{1-i/d}$ $\pargate_{n^{1/d}}$ gates, each of which takes in inputs from $n^{1/d}$ gates below it. Hence we observe a ``sharp threshold'' behavior where a difference in constants can change an exponential lower bound to a sublinear one.\

We also observe that this lower bound almost matches the classic $2^{\Omega(n^{\frac{1}{d-1}})}$ construction for $\ac_d$ circuits calculating parity. Thus, augmenting $\ac$ with unbounded fan-in gates which have the power to calculate the majority of polynomially many bits has \emph{no effect} on its ability to calculate parity, even thought we know such gates require exponentially sized $\ac_d$ circuits. In fact, by an argument resembling Shannon's classic circuit lower bound, there exists gates in $\g(k)$ which require size $2^{\Omega(n^{1/2d})}$ $\ac_d$ circuits.\ 

We can also show that the size lower bound in \Cref{cor:lb} and the correlation bound in \Cref{thm:paritycorr} is tight. In particular, the gap between the $2^{\Omega(n^{\frac{1}{d-1}})}$ lower bound for $\ac$ and $2^{\Omega(n^{\frac{1}{d}})}$ bound established for $\gt^0(.1n^{1/d})$ cannot be bridged. We defer the formal proofs to \Cref{sec:tight}.

\subsection{Correlation Bounds for $\gt^0(k)$ Circuits With Few Arbitrary Threshold Gates}
\label{sec:arbgate}

In this section, we prove that state of the art correlation bounds against $\ac$ circuits \cite{st19} with a small number of threshold gates extends to if we instead start with $\gt^0(\log^2 n)$ circuits. We first give an overview of their proof. As in previous works studying this correlation \cite{rw93,v07,ls11,st19}, the hard function we uncorrelate with is \[\rw_{m,k,r}(x) = \bigoplus_{i=1}^m\bigwedge_{j=1}^k\bigoplus_{\ell = 1}^r x_{ijk}\] A uniform string can be sampled by performing a random restriction and then filling the $\star$s with uniform bits. Driven by this, the overlying strategy is to apply a random restriction, and show the circuit collapses while the $\rw$ function maintains integrity. It turns out because our multi-switching lemma gives no loss in parameters (up to constants), we can apply the exact same argument in \cite{st19}, except replace the $\ac$ simplification lemma (Corollary 3.2 of \cite{st19}) with our more general \Cref{thm:constdepth}.


\begin{thm}
\label{thm:anycorr}
    Fix $u$. Let $v=.005\log n$ and $q=\sqrt{n/(v+1)}$ There exists a function $\rw_{q,v,q}\in \P$ and small enough constant $\tau$ such that for all circuits  $\any_u\circ \thr \circ \gt_d^0(\Omega(\log^2 n))$ circuits $F$ where each of the $u$ $\thr \circ \gt_d^0$ subcircuits of $F$ has size at most $s = n^{\tau \log n}$, we have  \[|\E_{x\sim U_n}[(-1)^{\rw(x) + C(x)}]\le 2^{-\Omega(n^{.499}/u)}\]
\end{thm}

\begin{proof}
We immediately apply \Cref{thm:constdepth} to $F$ with $G$ being the top $\any_u\circ \thr$ circuit, $ m=u\cdot 2^{(\eps/100d)\log^2 n}, w=\eps \log m, t = q/2$ and $k = (\eps/100d)\log^2 n$ to get that for $\rho'\sim R_p$, where $p = n^{-\eps/50}$, \[\Pr[F|_{\rho'}\text{ is an }(m/2,\any_u\circ \thr \circ \dt_w)\text{-decison tree}]\le 1 - 4d\cdot 2^{-q/2}. \] This computational model is now void of $\g(k)$ gates, and we can essentially port in the rest of \cite{st19} to finish. By the ``Second Step'' and ``Third Step'' under Section 2 of \cite{st19}, one can compose $\rho'$ with another restriction to get a final random restriction $\rho$ that simplifies the tree further and prunes the fan-in to an $\any_u\circ \thr \circ \AND_v$ circuit. 

It was shown in Lemma 4.3 of \cite{st19} that the same random restriction $\rho$ will have $\rw|_\rho$ equal(after restricting additional bits and negating input bits and/or the output) $\gip_{q/2, v+1}$ except with probability $2^{-\tilde{\Omega}(pq)}$. Theorem 21 in \cite{st19} then states that an $\any_u\circ \thr \circ \AND_v$ circuits can be calculated by a randomized NOF $(v+1)$-party protocol with error $\gamma = 2^{-q^{.99}/u}$ using $O(uv^3\log n\log (n/\gamma)) = O(q^{.99}v^3\log n)$ bits. Finally, by Theorem 14 in \cite{st19}, we can conclude the correlation between $\gip_{q/2, v+1}$ and $\any_u\circ \thr \circ \AND_v$ is at most $2^{-\Omega(q^{.99}/u)}$. Hence the overall correlation can be bounded, via union bound, by the sum of the error probabilities and the correlation of $\gip_{q/2, v+1}$ and $\any_u\circ \thr \circ \AND_v$, yielding  

\[|\E_{x\sim U_n}[(-1)^{\rw(x) + C(x)}]\le 4d\cdot 2^{-q/2} +2^{-\tilde{\Omega}(pq)} +  2^{-\Omega(q^{.99}/u)} = 2^{-\Omega(n^{.49}/u)}\] as desired.
\end{proof}

With this theorem, we can prove the actual correlation bound for $\gt^0(k)$ circuits with arbitrary gates.

\begin{thm}
    Let $C$ be a $\gt^0(\Omega(\log^2 n))$ circuit, $g$ of whose gates are arbitrary $\thr$ gates. Then \[\E[(-1)^{C(x) + \rw(x)}]\le 2^{-\Omega(\frac{n^{.499}}{g} - g)}.\] In particular, plugging in $g = \Theta(n^{.249})$ tells us \[\E[(-1)^{C(x) + \rw(x)}]\le 2^{-\Omega(n^{.249})}\]
\end{thm}

\begin{proof}
    This follows from \Cref{thm:anycorr} exactly like how Theorem 3 follows from Lemma 6 in \cite{ls11}.
\end{proof}

\begin{remark}
    We note that an argument analogous to the above can be used to show $2^{-\Omega(n^{.499})}$ correlation bounds against $\gt^0(\Omega(\log^2 n))$ circuits with $n^{.499}$ gates, via the same argument presented in \cite{st19}.
\end{remark}

It is worth noting that if we had tried performing this argument by expanding the size $n^{\Omega(\log n)}$ $\gt^0(\log^2 n)$ circuit naively into an $\ac$ circuit, not only would we get a loss in parameters, but the argument \emph{will not go through}. The proof crucially relied on correlation bounds against $v=.005\log n$ party protocols. Had we asymptotically increased the size of our circuit by writing it as an $\ac$ circuit, then after applying random restrictions to prune the fan-in of our circuit, we will be left with trying to uncorrelate against arbitrary $\omega(\log n)$-party protocols, a longstanding open problem (Problem 6.21 in \cite{kn96}).

\subsection{Derandomizing the Multi-Switching Lemma and PRGs for $\gt^0(k)$}
\label{sec:prg}

Using the same techniques appearing in  \cite{kel21,Lyu}, we can completely derandomize our switching and multi-switching lemma. We defer the proof to the appendix (\Cref{apx:demulti}).

\begin{thm}
\label{thm:demulti}
Let ${\cal F} = \{F_1,\dots, F_m\}$ be a list of size $m$ $\g(k)\circ \{\AND,\OR\}_w$ circuits.
     Let $(\Lambda, z)$ be a joint random variable such that 
     \begin{itemize}
         \item $\Lambda$ is a $(t+w)$-wise $p$-bounded subset of $[n]$
         \item Conditioned on any instance of $\Lambda$, $z$ $\eps$-fools CNF of size $\le m^2$.
     \end{itemize}
     Then  
     \[\Pr_{\Lambda, z}[{\cal F}|_{\rho(\Lambda, x)}\text{ has no $r$-partial depth-$t$ }\dt] \le 4(m2^k)^{t/r}(64pw)^t + (64wm)^{t+w}(2m)^{2kt/r}\cdot \eps\]
\end{thm}

\begin{proof}
    See the \Cref{apx:demulti} in the appendix.
\end{proof}

Using the derandomized multi-switching lemma, we can use the partition-based template in order to create PRGs for $\gt^0(k)$. The arument to reduce any constant depth to depth 2 will be a very similar argument. \cite{Lyu} simply uses a CNF PRG to tackle the base case of $\ac_2$ circuits, but we cannot do so with a $\g(k)\circ \AND_w$ circuit unless we want to expand it out as a CNF and incur a multiplicative $k$ loss in our seed length. We instead use the derandomized switching lemma one more time to simplify $\gt{(k)}\circ \AND_w$ to a $\dt_{\log m}$ and then fool this with an $(\eps/m, \log m)$-wise independent seed. We quickly prove this latter statement in the following lemma.

\begin{lem}[PRG for Depth $t$ Decision Trees]
\label{lem:basecase}
There exists a $\eps$-error PRG with $O(\log \log n + t + \log(1/\eps))$ seed length for $\dt_t$.
\end{lem}

\begin{proof}
Let $D$ be a $(\eps/2^t, t)$-wise independent distribution, samplable using $O(\log\log n  + t + \log(2^t/\eps)) = O(\log n + t + \log(1/\eps))$ bits. For arbitrary $T\in \dt_t$, label the leaves $L_1,\dots, L_{2^t}$, and let the value of leaf $L_i$ be $\ell_i$. Then 
\[ T(x) = \sum_{j=1}^{2^t}\ell_j\cdot \mathbbm{1}(T(x)\text{ reaches }L_j). \]

Note that $\ell_j\cdot \mathbbm{1}(T(x)\text{ reaches }L_j)$ depends on at most $t$ bits, and so $D$ will $\eps/2^t$-fool it. Therefore, \begin{align*}|\E_{x\sim U_n}[T(x)] - \E_{x\sim D}[T(x)]| & \le \sum_{j=1}^{2^t}|\E_{x\sim U}[\ell_j\cdot \mathbbm{1}(T(x)\text{ reaches }L_j) ]-\E_{x\sim D}[\ell_j\cdot \mathbbm{1}(T(x)\text{ reaches }L_j)]| \\ & \le \sum_{j=1}^{2^t}\eps/2^t \\ & = \eps. \end{align*}
\end{proof}

With this, we are now ready to prove our final PRG for $\gt^0(\log m)$.

\begin{thm}
\label{thm:prg}
For $m,n\in \N$ and $w\le \log m$, there is an $\eps$-error PRG with $O((w\log^{d-1}(m) + \log^2(m))\log(m/\eps)\log\log m)$ seed length for $\gt_d^0(\log m)\circ \AND_w$ circuits.
\end{thm}

\begin{proof}
Let $\ell =512w$ and $t = 10\log(m/\eps)$.
\begin{itemize}
    \item Let $H: [n]\to [\ell]$ be a $2t$-wise independent hash function which needs \[O(t\log n) = O(\log n\log(m/\eps))\] bits. We will let $H_i$ be an $n$-bit string such that $(H_i)_j = 1$ iff $H(j) = i$.
    \item Let $\eps' = \eps/(\ell\cdot 2^{t+1})$ and set $X_1,\dots, X_\ell$ to be strings that $\eps'$-fool $\gt_{d-1}^0(\log m)\circ\AND_{\log m}$ circuits of size $4m^2$ if $d\ge 2$, which by the inductive hypothesis uses \[O(\log^{d-1}(m)\log(m/\eps) \log \log m)\] seed length per $X_i$. If $d=1$, use the PRG from \Cref{lem:basecase} giving a seed length, which needs $ O(\log (m/\eps))$ seed per $X_i$.
    \item Let $Y$ be a string that $\eps/((64mw)^{t+w+1}(2m)^{2t})$-fools CNF of size $m^2$, samplable using \[O(\log m \log((mw)^{t+w}m^{t}/\eps)\log\log m) = O( \log(m/\eps) \log^2 m\log\log m)\] bits
\end{itemize}

The PRG will sample the above strings and output the following computation \[Y\oplus (X_1\wedge H_1)\oplus \cdots \oplus (X_\ell\wedge H_\ell) \] where $\wedge$ and $\oplus$ are the bitwise AND and XOR operations, respectively. Therefore, we get a total seed length of

 \begin{align*} O(\log n\log(m/\eps) + \ell \log (m/\eps) +  \log(m/\eps)\log^2 m\log\log m)  = O(\log(m/\eps)\log^2 m\log\log m)\end{align*} if $d=1$ and 
\begin{align*} O(\log n\log(m/\eps) + \ell \log^{d-1}(m)\log(m/\eps)\log\log m + & \log(m/\eps)\log^2 m\log\log m)  \\ & = O((w\log^{d-1}(m) + \log^2(m))\log(m/\eps)\log\log m).\end{align*}

Let $C$ be an arbitrary $\gt_d^0(\log m)\circ \AND_w$ circuit, and let $U_1,\dots U_\ell$ be independent and uniform $n$-bit strings. Like in \cite{Lyu} we use a hybrid argument to prove the theorem using the hybrid distributions \begin{align*} D_i = Y\oplus \bigoplus_{1\le j\le i}(U_i\wedge H_i)\oplus \bigoplus_{i < j\le \ell}(X_i\wedge H_i)\end{align*} for $0\le i\le \ell$. Noting $D_0$ is the PRG output, while $D_{\ell}$ is a uniform string, it suffices to show \begin{align}|\E_{x\sim D_{i-1}}[C(x)] - E_{x\sim D_{i}}[C(x)]|\le \eps/\ell\label{eqn:hybrid}\end{align} for all $1\le i\le \ell$, from which summing over all $i$ and applying the triangle inequality gets the desired result.\

Notice each $H_i$ is $2t$-wise $\frac{1}{\ell}$-bounded. Conditioned on $H$, note that $Z_i := Y\oplus \bigoplus_{1\le j < i}(U_i\wedge H_i)\oplus \bigoplus_{i < j\le \ell}(X_i\wedge H_i)$ $\eps/((64mw)^{t+w+1}(2m)^{2t})$-fools CNF of size $m^2$ since $Y$ does. Let ${\cal F}$ be the collection of all bottom depth-2 $\g(k)\circ \{\AND_w, \OR_w\}$ subcircuits of $C$. Therefore, if we let ${\cal E}$ be the event that ${\cal F}|_{\rho(H_i, Z_i)}$  has no $\log m$-partial depth-$t$ $\dt$, by \Cref{thm:demulti} it follows

\begin{align*}\Pr_{H,Y,U_{<i},X_{>i}}[{\cal E}] & \le 4(m2^{\log m})^{t/\log m}(64w/\ell)^t + (64mw)^{t+w}(2m)^{2 t\log m/\log m}\cdot \frac{\eps}{(64mw)^{2t}(2m)^{2t}} \\ & \le 4\cdot 2^{2t}(1/16)^t+ \frac{\eps}{(64mw)^{40\log(m/\eps)-\log m}} \\ & \le 4(1/4)^t + \frac{\eps}{4\ell} \\ & \le \frac{\eps}{2\ell} \end{align*}

Conditioning on $\neg {\cal E}, H,Y,U_{<i},X_{>i}$, we see  upon replacing all depth $2$ subcircuits with $\dt_t$s and applying \Cref{lem:depthred}, $C|_{\rho(H_i, Z_i)}$ is computable by a depth-$t$ $\dt$ where each leaf $L_j$ is an $\gt_{d-1}^0(k)\circ \{\AND,\OR\}_w$ circuit of size $\le m\cdot 2^{\log m} + m\le 2m^2$, and will become a $\dt_{\log m}$ tree if $d=1$. In either case, we see by construction that $X_i$ will fool it. Formally, we note that conditioned on the good events above, we have $C_{\rho(H_i,Z_i)}(y) = \sum_{j=1}^{2^t}L_j(y)\cdot\mathbbm{1}\{T(y)\text{ reaches }L_j\}$. By construction of $X_i$,\[|\E_{X_i}[L_j(X_i+Z_i)\cdot\mathbbm{1}\{T(X_i+Z_i)\text{ reaches }L_j\}] - \E_{U_i}[L_j(U_i+Z_i)\cdot\mathbbm{1}\{T(U_i+Z_i)\text{ reaches }L_j\}] | \le \frac{\eps}{\ell\cdot 2^{t+1}}\] Summing over all $1\le j\le 2^t$, applying the Triangle Inequality, and using linearity of expectation, we see \[\E_{X_i}[C_{\rho(H_i, Z_i)}(X_i+Z_i)] - \E_{U_i}[C_{\rho(H_i,Z_i)}(U_i+Z_i)]|\le \frac{\eps}{2\ell}.\]

Therefore, 

\[|\E_{x\sim D_{i-1}}[C(x)] - \E_{x\sim D_i}[C(x)]|\le \Pr[{\cal E}] + \Pr[\neg {\cal E}]\cdot \frac{\eps}{2\ell}\le \frac{\eps}\ell\] and (\ref{eqn:hybrid}) is proven.

\end{proof}

From this, we immediately get PRGs for size-$m$ $\gt_d^0(\log m)$ circuits.

\begin{thm}
   For every $m,n,d$ and $\eps > 0$, there is an $\eps$-PRG for size-$m$ $\gt_d^0(\log m)$ with seed length $O((\log^{d-1}(m) + \log^2(m))\log(m/\eps)\log\log m)$
\end{thm}

\begin{proof}
    Add trivial fan-in 1 gates to the bottom so that we effectively have a $\gt_d^0(\log m)\circ\AND_1$ circuit. By \Cref{thm:prg}, we can fool this with seed length $O((\log^{d-1}(m) + \log^2(m))\log(m/\eps)\log\log m).$
\end{proof}

\subsection{Fourier Spectrum Bounds for $\gt^0(k)$}
\label{sec:fourier}

Linial, Mansour, Nisan, and Tal showed that many notions of the Fourier spectrum of a function class is intimately related \cite{lmn93,man92,tal17}. \cite{tal17} writes out four key properties and conveniently describes the implications existing between them. We report a slightly altered version here.

\begin{thm}[\cite{lmn93,man92,tal17}]
\label{thm:equivalence}
Say for a class of functions, $\calC$ we have the following property.
\begin{itemize}
    \item \textbf{ESFT}: Exponentially small Fourier tails. For all $f\in\calC$,  \[W^{\ge k}[f]\le C e^{-\Omega(k/t)}.\] for some constant $C$.
\end{itemize}
 Then, $\calC$ also satisfies the following for some constant $C'$.
\begin{itemize}
    \item \textbf{SLPT}: Switching lemma type property. For all $ f\in\calC, d, p$, \[\Pr_{\rho\sim R_p}[\deg(C|_\rho)\ge d]\le C'\cdot O(pt)^d.\]

    \item \textbf{InfK}: Bounded total degree-$k$ influence. For all $f\in\calC$,$0\le k\le n$, \[\text{Inf}^k[f]\le C'\cdot O(t)^k.\]
    \item \textbf{L1}: Bounded $L_1$ norm at the $k$th level. For all $f\in\calC$, $0\le k\le n$, \[\sum_{|S| = k}|\widehat{f}(x)|\le C'\cdot O(t)^k\]

    \item \textbf{FMC}: Fourier mass concentration. For all $f\in\calC$, $f$ is $\eps$-concentrated on $t^{O(t\log (1/\eps))}$ coefficients.
\end{itemize}
\end{thm}

Due to the above unification result, it appears like we can bootstrap \Cref{cor:gt0} to give us a plethora of information about the Fourier spectrum of $\gt^0(k)$. Unfortunately, upon closer inspection, the Corollary doesn't quite give the exact property of SLPT. We instead show that $\gt^0(k)$ has ESFT. Our proofs will use the following lemma.

\begin{lem}[\cite{lmn93}]
\label{lem:restail}
For $f:\{\pm 1\}^n\to\{\pm 1\}$, $0\le \ell\le n$, and $p\in [0,1]$, \[W^{\ge \ell}[f]\le 2\E_{\rho\sim R_p}W^{\ge kp}[f|_\rho]\]  
\end{lem}

We first start off with depth 2 circuits.

\begin{lem}
\label{lem:spectrumbasecase}
    Let $f$ be a $\g(k)\circ \{\AND,\OR\}_w$. Then \[W^{\ge \ell}[f] \le 2\cdot 2^{-\ell/80w + k}\]
\end{lem}

\begin{proof}
    Let $p = 1/40w$ and $t = \ell/80w$. By \Cref{thm:switch}, if $\rho\sim R_p$, $f|_\rho$ becomes a depth-$t$ $\dt$ with $\ge 1-(20w/40w)^t2^k = 1-2^{-t+k}$ probability. Such trees have no Fourier mass above level $t$. Say $\rho$ is good if $f|_\rho$ does indeed become a $\dt$. Using \Cref{lem:restail} it follows \begin{align*}W^{\ge \ell}[f] &\le 2\E_{\rho\sim R_{p}}[W^{\ge p\ell}[f|_\rho]] \\ &\le 2\E_{\rho\sim R_{p}}[W^{\ge \ell/40w}[f|_\rho]|\rho\text{ is good}] + 2\cdot 2^{-t+k} \\ & \le 2\cdot 2^{-\ell/80w+k}.\end{align*}
\end{proof}

We can now use this as a base case to prove ESFT for $\gt^0$. We will need to utilize the following lemma.

\begin{lem}[\cite{tal17}]
Let $f:\{\pm 1\}^n\to\{\pm 1\}$, $0\le \ell\le n$, and let $T$ be a depth $d$ decision tree such that for any leaf $\ell$ and the corresponding restriction $\rho_\ell$ induced by the root-to-leaf path, we have $W^{\ge \ell}[f|_{\rho_\ell}]\le \eps$. Then $W^{\ge \ell + d}[f]\le \eps$.
\end{lem}

We now state and prove the theorem. Define the \emph{effective size} of a Boolean circuit to be the number of gates in the circuit at distance 2 or more from the inputs.

\begin{thm}
\label{thm:esft}
    Let $f$ be a $\gt_d^0(k)\circ \{\AND,\OR\}_w$ circuit with effective size $m$. Then \[W^{\ge \ell}[f] \le 4^d\cdot 2^{-\frac{\ell}{80w(128(k+\log m))^{d-1}} + k} \]
\end{thm}

\begin{proof}
    We apply induction. The base case of $d=1$ is taken care of by \Cref{lem:spectrumbasecase}. \ 
    
    We now prove the inductive step for depth $d$. Sample $\rho\sim R_p$ with $p = \frac{1}{128w(m2^k)^{1/(k+\log m)}} = \frac{1}{128w}$, and let $t = p\ell/2 = \ell/256w$. By \Cref{lem:multiswitch}, all the bottom depth-2 $\g(k)\circ\{\AND,\OR\}_w$ subcircuits of $f|_\rho$ can be calculated by a $(k+\log m)$-partial depth-$t$ decision tree with probability $\ge 1-4\cdot 2^{-t}$. By \Cref{lem:depthred}, this implies $f|_\rho$ becomes a $(t,\gt_{d-1}^0(k)\circ\{\AND,\OR\}_{k+\log m})$-tree $T$. Furthermore, each leaf circuit has effective size $\le m$. Call $\rho$ \emph{good} if $f|_\rho$ simplifies to such a tree. Then

    \[
        W^{\ge \ell}[f] \le 2\E_{\rho\sim R_{p}}[W^{\ge p\ell}[f|_\rho]] \le 2\E_{\rho\sim R_{p}}[W^{\ge p\ell}[f|_\rho]|\rho\text{ is good}] + 8\cdot 2^{-t}. 
    \]

Fix a good $\rho$. For a leaf $L$ of $T$, let $\tau_L$ be the restriction induced by the path to $L$ in $T$. We know by Lemma (cite) that \[W^{\ge p\ell}[f|_\rho]\le \max_{\text{leaf }L}W^{\ge p\ell - t}[f|_{\rho\circ\tau_L}]\le \max_{\text{leaf }L}W^{\ge p\ell/2}[f|_{\rho\circ\tau_L}].\] As $f|_{\rho\circ\tau_L}$ is a $\gt_{d-1}^0(k)\circ\{\AND,\OR\}_{k+\log m}$ circuit for every $L$ we can then use the inductive hypothesis to bound \[\max_{\text{leaf }L} W^{\ge p\ell/2}[f|_{\rho\circ\tau_L}]\le 4^{d-1}\cdot 2^{-\frac{p\ell}{80(k+\log m)(128(k+\log m))^{d-2}} + k} = 4^{d-1}\cdot 2^{-\frac{\ell}{80w(128(k+\log m))^{d-1}} + k}.\] Putting this all together, we get

\begin{align*}
W^{\ge \ell}[f] & \le 2\E_{\rho\sim R_{p}}[W^{\ge p\ell}[f|_\rho]|\rho\text{ is good}] + 8\cdot 2^{-t} \\ 
& \le 2\cdot 4^{d-1}\cdot 2^{-\frac{\ell}{80w(128(k+\log m))^{d-1}} + k} + 8\cdot 2^{-\ell/256w} \\ & \le 4^d\cdot 2^{-\frac{\ell}{80w(128(k+\log m))^{d-1}} + k}
\end{align*}

\end{proof}

We can bootstrap \Cref{thm:esft} with \Cref{thm:equivalence} to yield the following properties about $\gt^0(k)$

\begin{thm}
\label{thm:gtfourier}
Let $f$ be a size-$m$ $\gt_d^0(k)$ circuit and define $t:= (k+\log m)^{d-1}$. Then the following is true for some $C$

    \begin{enumerate}
    \item ESFT: $W^{\ge \ell}[f]\le C\cdot 2^k \cdot 2^{-\Omega\left(\frac{\ell}{t}\right)}.$
    \item SLTP: For all $0<p<1$, $\Pr_{\rho\sim R_p}[\deg(f|_\rho)\ge \ell]\le C\cdot  O(pkt)^\ell.$

    \item InfK: $\text{Inf}^\ell[f]\le  C\cdot O(kt)^\ell.$
    \item L1: $\sum_{|S|= \ell} |\widehat{f}(x)|\le C\cdot  O(kt)^\ell.$
    \item FMC: $f$ is $\eps$-concentrated on $2^{O((k+\log(1/\eps))t\log t)}$ coefficients.
\end{enumerate} where $f$ and any hidden constants only depend on $d$.
\end{thm}

\begin{proof}
Add a trivial $(d+1)$-st layer of $\AND_1$ gates at the base of $f$ and apply \Cref{thm:esft} to deduce that \[W^{\ge \ell}[f]\le 4^d\cdot 2^{-\frac{\ell}{80(128(k+\log m))^{d-1}} + k}, \] proving the first item. Now since we know $W^{\ge \ell}[C]\le 1$ (by Parseval's) and $k\ge 1$, it follows that \[W^{\ge \ell}[f]\le (W^{\ge \ell}[C])^{1/k}\le C_d \cdot 2^{-\Omega\left(\frac{\ell}{kt}\right)}. \] Therefore, the second, third, and fourth items follow by applying \Cref{thm:equivalence} (as well as a version of the fifth item with weaker parameters). We now prove Item 5. \\

Notice that for $w:= t\cdot O(k +\log(1/\eps))$, we have by Item 1 that $W^{\ge w}[f]\le \eps/2$. Now by Item 4, \begin{align}\sum_{|S|< w} |\widehat{f}(S)| \le \sum_{i=0}^{w-1} O(kt)^i \le (C'kt)^w. \label{eqn:l1bound}\end{align} Now let $\calF = \{S : |S| < w\text{ and } |\widehat{f}(S)|\ge \frac{\eps/2}{(C'kt)^w}\} $. Notice that \begin{align*}\sum_{S\in \calF} \widehat{f}(S)^2 & = 1 - \sum_{|S|\ge w} \widehat{f}(S)^2 - \sum_{|S| < w, S\notin\calF}\widehat{f}(S)^2 \\ & \ge 1-\eps/2 - \frac{\eps/2}{(C'kt)^w}\sum_{|S| < w}|\widehat{f}(S)| \\ & \ge 1-\eps.\end{align*} By \cref{eqn:l1bound}, the maximum number of terms in $\calF$ can be at most \[(C'kt)^w/\left(\frac{\eps/2}{(C'kt)^w}\right) = 2(C'kt)^{2w}/\eps = 2^{O(w\log(kt) + \log(1/\eps))} = 2^{O((k+\log(1/\eps))t\log t)}\] and thus Item 5 is proved.
\end{proof}

As a first application, the work of Kushilevitz and Mansour \cite{km91} allows us to translate FMC to learnability results.

\begin{lem}[\cite{km91}]
\label{lem:learn}
    Let $f$ be a Boolean function such that there exists a $t$-sparse multivariate polynomial $g$ (over the Fourier basis) such that $\E_{x\sim U_n}[(f(x)-g(x))^2]\le \eps$. There exists a randomized algorithm, whose running time is polynomial in $t,n, 1/\eps, \log(1/\delta)$ such that given blackbox access to $f$ and $\delta > 0$, outputs a function $h$ such that over the randomness of the algorithm, \[\Pr[\E_{x\sim U_n}[(f(x)-h(x))^2]\le O(\eps)]\ge 1-\delta.\]
\end{lem}

Using this lemma, we can derive a learning algorithm for $\gt^0(k)$.

\begin{thm}
    There exists an algorithm such that given blackbox access to any $C\in\gt_d^0(k)$ of size $m$ and $\delta>0$, outputs a function $h$ such that over the randomness of the algorithm, \[\Pr[\E_{x\sim U_n}[(f-h)^2]\le O(\eps)]\ge 1-\delta.\] Furthermore, this algorithm runs in $\text{poly}(n, 2^{\tilde{O}((k+\log(1/\eps))(k+\log m)^{d-1})}, 1/\eps , \log(1/\delta))$
\end{thm}

\begin{proof}
    From \Cref{thm:gtfourier}, for any $C\in \gt_d^0(k)$, there exists $g$ of sparsity $t = 2^{\tilde{O}((k+\log(1/\eps))(k+\log m)^{d-1})}$, created by taking the Fourier expansion of $C$ and only keeping the $\eps$-concentrated coefficients $\calS\subset 2^{[n]}$, such that \[\E_{x\sim \{\pm 1\}^n}[(C(x)-g(x))^2]\le \E_{x\sim \{\pm 1\}^n}\left[\left(\sum_{S\notin \calS} \widehat{C}(S)x^S\right)^2\right] = \sum_{S\notin \calS} \widehat{C}(S)^2\le \eps.\] The result then follows by \Cref{lem:learn}.
\end{proof}

We also can prove a new correlation bound result with this Fourier spectrum. It is known that $\maj$ is a symmetric function that has $O_d(\log^{d-1}(m)/\sqrt{n})$ correlation against size-$m$ $\ac_d[\oplus]$ circuits. A natural question to ask is whether Majority is special in this regard, or if a random symmetric function (use $n+1$ coin tosses to assign a bit to each Hamming level) will display $O_d(\log^{d-1}(m)/n^{\alpha})$ correlation against size-$m$ $\ac_d[\oplus]$ circuits for some $\alpha$. Tal (\cite{tal17}, Theorem 6.1) used ESFT and L1 of $\ac$ to prove that random symmetric functions (or more specifically balanced symmetric functions) display $O_d(\log^{d-1}(m)/\sqrt{n})$ correlation against size-$m$ $\ac_d$ circuits, so it is natural to believe that this should similarly be true against $\ac[\oplus]$. Unfortunately, since $\pargate$ has all its Fourier weight at level $n$, this proof approach is doomed to fail for $\ac[\oplus]$ circuits, as the class doesn't demonstrate ESFT. However, we can now give partial progress towards this goal by showing a random symmetric function uncorrelates with $\gt^0(k)$ circuits, as this class contains gates which calculate parity as long as the Hamming weight of the input is at most $k$. This result can be seen as finding out how general of a circuit class we can stretch the Fourier argument before we reach the roadblock on this approach demonstrated by $\pargate$.

\begin{thm}
Let $f\in \gt_d^0(k)$, and let $g$ be a symmetric function, both mapping $\{\pm 1\}^n\to\{\pm 1\}$, and let $(k+\log m)^{d-1}\le O\left(\frac{n}{k+\log n}\right)^{1/3}$. Then \[\corr(f,g) := \E_x[f(x)g(x)] \le |\widehat{g}(\emptyset)| +  \frac{C_dk(k+\log m)^{d-1}}{\sqrt{n}}\] 
\end{thm}

\begin{proof}
    We note for $\ell'$ to be picked later, we can decompose \begin{align}\corr(f,g) = |\E_x[f(x)g(x)]| = \left|\sum_{S\subset[n]} \widehat{f}(S)\widehat{g}(S)\right|\le |\widehat{g}(\emptyset)| + \sum_{|S| < \ell'}|\widehat{f}(S)\widehat{g}(S)| + \sum_{|S|\ge \ell}|\widehat{f}(S)\widehat{g}(S)|.\label{eqn:corr}\end{align} We will bound the first summation using L1, and the second summation by ESFT. The second summation can be bounded as follows using Cauchy-Schwarz. \begin{align}\sum_{|S|\ge \ell}|\widehat{f}(S)\widehat{g}(S)|\le \sqrt{W^{\ge \ell'}[f]\cdot W^{\ge \ell'}[g]}\le \sqrt{ 4^d\cdot 2^{-\frac{\ell'}{80(128(k+\log m))^{d-1}} + k}}\le 1/\sqrt{n}\end{align}

    if we set $\ell' = c_d(k+\log n)(k + \log m)^{d-1}$ for some constant $c_d$ only depending on $d$. Now to bound the first summation, note since $g$ is symmetric, $\widehat{g}(S)$ is constant over all $S$ of same cardinality. Therefore, \[|\widehat{g}(S)| = \sqrt{\widehat{g}(S)^2} = \sqrt{\frac{1}{\binom{n}{|S|}}\sum_{S'; |S'|=|S|} \widehat{g}(S')^2}\le \sqrt{\frac{1}{\binom{n}{|S|}}}.\] Hence using L1 from \Cref{thm:gtfourier}, we can bound \begin{align}\sum_{|S| < \ell'}|\widehat{f}(S)\widehat{g}(S)|\le \sum_{1\le \ell < \ell'}\sqrt{\frac{1}{\binom{n}{\ell}}}\sum_{|S| = \ell}|\widehat{f}(S)|\le \sum_{1\le \ell < \ell'}\left(\frac{C_dk(k+\log m)^{d-1}}{\sqrt{n/\ell}}\right)^\ell\label{eqn:geoseries}\end{align} where $C_d$ is some constant depending on $d$. We bound this sum by a geometric series of the same first term and with common ratio $1/2$. Indeed, we see that the ratio of consecutive terms will be \[\frac{\left(\frac{C_dk(k+\log m)^{d-1}}{\sqrt{n/(\ell+1)}}\right)^{\ell+1}}{\left(\frac{C_dk(k+\log m)^{d-1}}{\sqrt{n/\ell}}\right)^\ell} =\frac{C_dk(k+\log m)^{d-1}}{\sqrt{n}}\sqrt{\frac{(\ell+1)^{\ell+1}}{\ell^\ell}}\le \frac{C_dk(k+\log m)^{d-1}}{\sqrt{n}}\sqrt{e\ell'} \le 1/2  \] where the last inequality follows from the assumption $(k+\log m)^{d-1}\le O\left(\frac{n}{k+\log n}\right)^{1/3}$. Hence the quantity in \Cref{eqn:geoseries} can be upper bounded by twice the first term, so \[\sum_{|S| < \ell'}|\widehat{f}(S)\widehat{g}(S)|\le \frac{C_dk(k+\log m)^{d-1}}{\sqrt{n}}.\] Hence from (\ref{eqn:corr}), \[\corr(f,g)\le |\widehat{g}(\emptyset)| + \frac{1}{\sqrt{n}} + \frac{C_dk(k+\log m)^{d-1}}{\sqrt{n}}.\]
\end{proof}
\section{Open Problems}

We conclude with some directions for future research.
\begin{itemize}
    \item Our tightness result in \Cref{apx:tight} uses a function in $\g(k)\circ \AND_w$, but it is not known whether a $k$-$\OR\circ \AND_w$ circuit can saturate the bound. In particular, is \Cref{thm:switch} tight for $\ac(k)$ or $\tc^0(k)$ circuits?
    \item It was already noted that \Cref{cor:lb} is tight in essentially every way possible for $\gt^0(k)$ circuits. However, all our tightness results (\Cref{apx:tight}) use constructions that abuse the generality of $\gt(k)$. Are there constructions exhibiting tightness which are in $\tc^0(k)$ or $\ac(k)$? Alternatively, can we obtain stronger size lower bounds if we were only concerned with $\ac(k)$ or even $\tc^0(k)$ circuits? Either finding a pathological construction in $\ac(k)/\tc^0(k)$ or proving stronger lower bounds for these weaker circuit class would be interesting.

    \item We touched up on how a result of Allender and Kouck\'{y} (\cite{ak10}, Theorem 3.8) states that there exists an absolute constant $C_{AK}$ such that $\maj_n$ can be written as an $\ac(n^\eps)$ circuit with depth $\le C_{AK}/\eps$ and size $O(n^{1+\eps})$. \cite{ak10} actually only use $\AND_2,\OR_2$, and $\maj_{n^\eps}$ gates. If we were allowed all the gate classes in $\ac(n^\eps)$, could we find a better construction (more specifically a lower depth blowup)? That way, we would get a stronger reduction from proving bounds on $\tc$ to $\ac(n^\eps)$ where we only need to show size lower bounds on smaller depth $\ac(n^\eps)$ circuits. 
    \item Are there any other applications of the generalized switching lemma? Due to the versatility of this theorem, it can essentially be plugged in wherever the classical switching lemma was used to get more general results. Perhaps this can generalize other results or even push a switching lemma argument that initially wouldn't go through (the remark after \Cref{apx:tight} gives an example of the general switching lemma giving stronger bounds than the classical one).
\end{itemize}

\section*{Acknowledgements}
The author thanks David Zuckerman and Chin Ho Lee for many valuable discussions, Xin Lyu for explaining his work in \cite{Lyu}, anonymous reviewers for valuable feedback and for pointing us to the construction presented in \Cref{apx:tightlb}, and Jeffrey Champion, Shivam Gupta, Michael Jaber and Jiawei Li for helpful comments.

\bibliographystyle{alpha}        
\bibliography{refdb.bib}

\appendix
\section{Deferred Proofs}

\subsection{Showing $\tc^0(k)\subset \gt^0(k)$}
Here, we prove that circuits created by biased $\ltf$ gates are indeed contained in $\gt^0(k)$.
\begin{thm}
\label{thm:thrgk}
Any $\thr$ gate $f$ with balance $\le k$ (see \Cref{defn:bal}) is, upon negating certain input bits, in $\g(k)$.
\end{thm}

\begin{proof}
    Let $f:\{0,1\}^n\to\{\pm 1\}$ be defined as $f(x) = \text{sgn}(\sum_{i=1}^n{w_i(-1)^{x_i}} - \theta)$. By negating input bits, we can assume each $w_i\ge 0$. Furthermore, since the definition of $\g(k)$ is symmetric (solely depends on the sum of input bits), we can assume WLOG that $0\le w_1\le\cdots\le w_n$. Since $f$ has balance $\le k$, we know that $-\sum_{i\le k} w_i + \sum_{i>k}w_i < \theta$. Assuming $-\sum_{i\le k} w_i + \sum_{i>k}w_i < |\theta|$, we will show $f$ is an orlike $\g(k)$ gate (an analogous proof will show the case for $-\theta$ and andlike). \ 
    
    Consider $x$ such that $\sum x_i\ge k$. Let $S_x\subset [n]$ denote the set where $i\in S_x$ iff $x_i = 1$. It follows that \begin{align}f(x) & = \text{sgn}\left(-\sum_{i\in S_x}w_i + \sum_{i\notin S_x} w_i - \theta \right) \nonumber \\ & = \text{sgn}\left(-\sum_{i\le k}w_i + \sum_{i> k} w_i - \theta  + 2\sum_{i\in [k]\setminus S_x} w_i- 2\sum_{i\in [k+1,n]\cap S_x}w_i  \right) \label{eqn:ltf}\end{align}

    We know by assumption that $-\sum_{i\le k}w_i + \sum_{i> k} w_i - \theta \le 0$. Since $|S_x|\ge k$, \[|[k]\setminus S_x| = k - |k\cap S_x| \le  |[k+1,n]\cap S_x|, \] and each element in $[k]\setminus S_x$ is strictly smaller than each element in $[k+1,n]\cap S_x$. Combining these observations with the fact $w_1\le\dots\le w_n$, it follows \[2\sum_{i\in [k]\setminus S_x} w_i- 2\sum_{i\in [k+1,n]\cap S_x}w_i \le 0.\]

    Therefore, \[-\sum_{i\le k}w_i + \sum_{i> k} w_i - \theta  + 2\sum_{i\in [k]\setminus S_x} w_i- 2\sum_{i\in [k+1,n]\cap S_x}w_i\le 0.\] Combining this with (\ref{eqn:ltf}) it follows $f(x) = -1$ for any $x$ with $\sum x_i\ge k$. Hence $f$ is an orlike $\g(k)$ gate as desired.
\end{proof}

\subsection{Tightness of the $\gt^0(k)$ Switching Lemma and Correlation Bounds}
\label{sec:tight}
In this section, we give constructions which show that various bounds we establish are indeed tight. We first show that the switching lemma we established is tight.

\begin{thm}
\label{apx:tight}
Let $p,w,t,k$ be parameters such that $pkw<1/2$ and $t > k/2$. There exists a $\g(k)\circ \AND_w$ circuit $C$ such that  \[\Pr_{\rho\sim R_p}[\dt(C|_\rho)\ge t]\ge 2^{k/2}(.5pw)^t\]
\end{thm}

\begin{proof}


We will take $C = \pargate_{kw}$. To see that this is computable in $\g(k)\circ \AND_w$, write \[\pargate_{kw}(x) = \pargate_k(\pargate_w(x_1,\dots, x_w),\dots, \pargate_w(x_{(k-1)w+1},\dots, x_{kw})).\] Now, write out each bottom layer $\pargate_w$ as a size $w2^{w-1}$ CNF which takes the $\OR$s of the $2^{w-1}$ $\AND$ clauses corresponding to $w$-bit inputs with an odd number of ones. Notice that for any assignment, 
\begin{itemize}
\item at most \emph{one} of the $2^{w-1}$ clauses under each $\OR$ can simultaneously be satisfied,
\item which implies at most $k$ of the bottom layer clauses can be simultaneously satisfied.
\end{itemize}
By the first bullet point, we can turn the $\OR$ gates into $\pargate$ gates, turning $C$ into a $\pargate_{k2^{w-1}}\circ \AND_w$ circuit. By the second bullet point, we can replace the top $\pargate_{k2^{w-1}}$ gate with the gate $\g\in \g(k)$ which calculates parity if at most $k$ input bits are one, and outputs $0$ otherwise. Consequently, $C$ can be calculated by a $\g(k)\circ \AND_w$ circuit. \ 

We can now directly calculate the simplification probability. Notice that if $\ge t$ variables are alive in $\rho$, then $C|_\rho$ must have decision tree depth $\ge t$ (since even if $d-1$ input bits are known, parity remains ambiguous). Hence we calculate \begin{align*}\Pr_{\rho\sim R_p}[\dt(C|_\rho)\ge t] & \ge \sum_{i=t}^{kw}\binom{kw}{i}p^i(1-p)^{kw-i} \\ & \ge (1-p)^{kw}\left(\frac{kw}t\right)^t\left(\frac{p}{1-p}\right)^t \\ & \ge (1-pkw)2^{k/2}\left(\frac{pw}{1-p}\right)^t \\ & \ge 2^{k/2}(.5pw)^t \end{align*} where the third inequality follows from the fact $(k/t)^t$ is increasing in $t$ for $t > k/2$
 \end{proof}

\begin{remark}
    Notice that if $\pargate_{n}$ was written as a width $n$-CNF and the classical switching lemma was applied, we would get a ``failure to simplify'' probability bound of $\le (5pn)^t$. However, if we rewrite $\pargate_n$ as a $\g(t\log(n/t))\circ \AND_{n/t\log(n/t)}$ circuit in a similar manner as above, and use \Cref{thm:switch}, we get a bound of $\le (20pn/t\log(n/t))^t2^{t\log(n/t)}\le (20pn^2/t^2)^t$, which is asymptotically stronger when $t = \omega(\sqrt{n})$. This shows that we can potentially obtain tighter parameters by expressing functions as a more compact $\gt^0(k)$ circuit and applying \Cref{thm:switch} rather than using the classical switching lemma on a larger $\ac$ circuit computing the same function.
\end{remark}

 We will now show that the circuit lower bound established in \Cref{cor:lb} is tight. We thank an anonymous reviewer for pointing us to this construction.

\begin{thm}
\label{apx:tightlb}
    There exists a size-$2^{O((n/k)^{\frac{1}{d-1}})}$ $\gt_d^0(k)$ circuit which computes $\pargate_n$.
\end{thm}

\begin{proof}
    Split the input string into $k$ blocks of size $n/k$ bits each. We can compute the parity of each of the $n/k$-size blocks straightforwardly using a depth $d-1$ circuit made out of $\pargate_{(n/k)^{\frac{1}{d-1}}}$ gates; iteratively group the bits into blocks of $(n/k)^{\frac{1}{d-1}}$, and use a gate to take the parity of each block, thereby creating a depth $d-1$ tree of $\pargate_{(n/k)^{\frac{1}{d-1}}}$ gates. Finally, we can take a $\pargate_{k}$ of these $k$ depth-$(d-1)$ circuits to get a depth $d$ circuit which computes the parity of all $n$ bits. We now focus on converting this parity-riddled circuit to a $\gt^0(k)$ one. \ 

    Consider the top depth-2 subcircuit, which is a $\pargate_{k}\circ \pargate_{(n/k)^{\frac{1}{d-1}}}$ circuit. Notice that $\pargate_{k}\in \g(k)$ and $\pargate_{(n/k)^{\frac{1}{d-1}}}$ is trivially computable by a decision tree in $\dt_{(n/k)^{\frac{1}{d-1}}}$. Therefore by \Cref{lem:depthred}, this top subcircuit can be replaced by a size $2^{O((n/k)^{\frac{1}{d-1}})}$ $\g(k)\circ \OR_{(n/k)^{\frac{1}{d-1}}}$. Consequently, we have converted our original depth-$d$ circuit into a new one where the first 2 layers are made from gates in $\g(k)$. \ 

    We now use the fact that any $\pargate_{(n/k)^{\frac{1}{d-1}}}$ can be expressed as a size $2^{O((n/k)^{\frac{1}{d-1}})}$ CNF or DNF to convert the remaining $\pargate_{(n/k)^{\frac{1}{d-1}}}$ gates to $\AND/\OR$ gates while preserving the depth. Convert the third layer $\pargate_{(n/k)^{\frac{1}{d-1}}}$ gates to a DNF ($\OR$ of $\AND$s), and then collapse the 2nd and 3rd layer as they both consist solely of $\OR$ gates. The third layer now consist of $\AND$ gates, so replace the fourth layer $\pargate_{(n/k)^{\frac{1}{d-1}}}$ gates with a CNF ($\AND$ of $\OR$s) to again induce a collapse of the consecutive $\AND$ layers. Repeat this procedure down to the bottom of the circuit. \ 

    Clearly after this procedure, all gates are in $\g(k)$ (in fact, all but the top gate are $\AND$ or $\OR$). Notice that at each stage, we increased the circuit depth by 1 when plugging in the DNF/CNF, but then reduced the circuit depth when collapsing layers of the same gate type. Hence the final circuit is still of depth $d$. We already calculated the top depth-2 subcircuit is of size $2^{O((n/k)^{\frac{1}{d-1}})}$. Assuming we didn't collapse any gates, we would have replaced each of the $ o(n^{1-\eps})$ parity gates with a size-$2^{O((n/k)^{\frac{1}{d-1}})}$ circuit, so clearly the final circuit size is at most \[2^{O((n/k)^{\frac{1}{d-1}})} + o(n/k)\cdot 2^{O((n/k)^{\frac{1}{d-1}})}\le 2^{O((n/k)^{\frac{1}{d-1}})}.\] Therefore at the end of this procedure, we get our desired circuit.
\end{proof}

We now show that not only is the size lower bound tight, but the average case correlation bound established in \Cref{thm:paritycorr} as well. 

\begin{thm}
 Assume $m\ge k^d$. There is a $\gt^0(k)$ circuit $C$ of size $\le m$ such that \[\E_x[(-1)^{C(x) + \pargate(x)}]\ge 2^{\Omega(k)}\cdot 2^{-O(n/(k+\log m)^{d-1})}\] 
\end{thm}

\begin{proof}
     Let $M=\max\{k,c_d\log m\}$, where $c_d$ is a constant such that the parity over $(c\log m)^{d-1}$ bits can be computed by an $\ac_d$ circuit of size $\le m$. Split the input into $\ceil{n/M^{d-1}}$ blocks of size $\le M^{d-1}$. If $M = c_d\log m$, each block can be calculated by an $\ac_{d}$ circuit of size $m$, and if $M=k$, each block can be calculated by a $\le k^d\le m$-size $\gt_{d-1}^0(k)$ circuit using a tree of $\pargate_k$s. Now if $\ceil{n/M^{d-1}}=1$, this circuit computes the parity of all $n$ bits and we are done, so assume $\ceil{n/M^{d-1}}\ge 2$. \ 
     
     Join all the $\ceil{n/M^{d-1}}$ subcircuits by the gate $G$ defined to compute parity if the Hamming weight of the input is at most $k$, and to equal 0 otherwise. Clearly $G\in\g(k)$. Therefore, if the subcircuits were constructed to be in $\ac_{d}$, we can collapse the top two layers into one using \Cref{lem:depthred} (similar to \Cref{apx:tightlb}), giving us a $\gt^0_d(k)$ circuit. If the subcircuits were $\gt_{d-1}^0(k)$, we trivially get a $\gt^0_d(k)$ gate after adding $G$. In the case more than $k$ of the $\ceil{n/M^{d-1}}$ input blocks have parity 1, our circuit will be constant, and thus will agree with parity about half the time. Let us crudely lower bound the correlation in this case to be $0$.
     When $\le k$ have parity 1, the top gate computes parity exactly. Therefore our correlation is simply the probability at most $k$ of the input blocks have parity $1$, which is simply \[2^{-\ceil{n/M^{d-1}}}\sum_{i\le k} \binom{\ceil{n/M^{d-1}}}k \ge 2^{-\ceil{n/M^{d-1}}}\cdot 2^{\Omega(k)}\ge 2^{\Omega(k)}\cdot 2^{-O(n/(k+\log m)^{d-1})}.\] This gives our correlation lower bound as desired.
\end{proof}

\subsection{Proof of the $\gt^0(k)$ Multi-Switching Lemma}

We prove the multi-switching lemma here.

\begin{thm}[Proof of \Cref{lem:multiswitch}]
\label{apx:multiswitch}
Let ${\cal F} = \{F_1,\dots, F_m\}$ be a list of 
$\g(k)\circ \AND_w$ circuits on $\{0,1\}^m$. Then\[ \Pr_{\rho\sim R_p}[{\cal F}|_\rho\text{ does not have $r$-partial depth-$t$ DT}]\le 4(64(2^km)^{1/r}pw)^t \]
\end{thm}
\begin{proof}
    We follow the proof in \cite{Lyu} exactly, where the only difference is that the ``Canonical Partial Decision Tree'' (CPDT) will use the modified CDT we created in \Cref{alg:one}, the definition of global witnesses (resp. global partial witnesses) will now use the witnesses (resp. partial witnesses) that we defined in \Cref{defn:wit} (resp. \Cref{defn:partwit}), and the ``Global Witness Searcher'' will run the modified witness searcher we created in \Cref{alg:two}.\

    Consider the following CPDT procedure.

\SetKwComment{Comment}{/* }{ */}
\SetKwBlock{init}{initialize:}{}
\SetKwInput{kwInput}{Input}
\SetKwInput{kwReturn}{return}
    \begin{algorithm}[hbt!]
    \caption{Canonical Partial Decision Tree}
    \label{algo:canonical-partial}
    \kwInput{ A list of $\g(k)\circ \{\AND,\OR\}_w$ circuits ${\cal F} = \{F_1,\dots, F_m\}$, black-box access to a string $\beta \in \{0,1\}^n$, and an auxiliary string $z\in \{0,1\}^n$.}
    \init{
        $x \gets (\star)^n$.\\
        $j \gets 1$. \\
        $\mathrm{counter} \gets 0$. \\
    }
    \While{$\mathrm{counter} < t$} {
        Find the smallest $i \ge j$ such that $\dt(F_i|_{x}) > w$. If no such $i$ exists, exit the loop. \\
        $y\gets (\star)^n$. \\
        $I\gets \emptyset$. \\
        \While{$F_i|_{x\circ y}(\star)$ is not constant and $\mathrm{counter} < t$} {
            $C_{i, q}\gets$ the term that $T_{F_i|_{x\circ y}}$ from \Cref{alg:one} will query. \\
            $B_{i, q}\gets $ the set of unknown variables in $C_{i, q}|_{x\circ y}$. \\
            $y_{B_{i, q}} \gets z_{B_{i, q}}$. \\
            $I\gets I \cup B_{i, q}$. \\
            $\mathrm{counter} \gets \mathrm{counter} + |B_{i, q}|$. \\
        }
        Query $\beta_{I}$, and set $x_I\gets \beta_I$. \\
        $j\gets i$. 
    }
    \textbf{return} {$x$} 
    \end{algorithm}

With this, we can define the following notion of a ``global witness'' to intuitively be a transcript on adversarially chosen inputs.

\begin{defn}
Let $t,w$ be two integers. Consider a list of $\g(k)\circ \{\AND,\OR\}_w$ circuits $\calF = \{F_1, \dots, F_m \}$. Suppose $\rho \in \{0,1,\star\}^n$ is a restriction. Let $(R, L_i, S_i, W_i, \beta_i)$ be a tuple, where
\begin{itemize}
    \item $1\le R \le \frac{t}{r}$ is an integer;
    \item $1\le L_1 \le L_2 \le \dots \le L_R\le m$ is a list of $R$ non-decreasing indices;
    \item $S_1.\dots, S_R$ is a list of $R$ integers such that $\sum_{i=1}^R S_i \in [t, t+w]$;
    \item $W_1,\dots, W_R$ is a list of witnesses (as per \Cref{defn:wit}). For every $i\in [R]$, $W_i$ has size $S_i$;
    \item $\beta_1,\dots, \beta_R$ are $R$ strings where $|\beta_i| = S_i$ for every $i\in [R]$.
\end{itemize}
We call the tuple a $(r,t)$-global witness for $\rho$, if it satisfies the following.
\begin{enumerate}
    \item Set $\rho_1 = \rho$. $W_1$ is a $S_1$-witness for $F_{L_1}|_{\rho_1}$.
    \item For every $i\ge 2$, let $I_{i-1}\subseteq [n]$ be the set of variables involved in $W_{i-1}$. Note that $|I_{i-1}| = S_{i-1}$ since the size of $W_{i-1}$ is $S_{i-1}$. Identify $\beta_{i-1}$ as a partial assignment in $\{0,1,\star\}^n$ where only the part $\beta_{i-1, I_{i-1}}$ is set and other coordinates are filled in with $\star$. Construct $\rho_i = \rho_{i-1}\circ \beta_{i-1}$. Then $W_i$ is a $S_i$-witness for $F_{L_i}|_{\rho_{i}}$.
\end{enumerate}
The size of the global witness is defined as $\sum_{i=1}^R S_i$.
\end{defn}

\begin{lem}
\label{lem:globw}
Consider a list of $\g(k)\circ \{\AND,\OR\}_w$ circuits $\calF = \{F_1, \dots, F_m \}$. Suppose $\rho \in \{0,1,\star\}^n$ is a restriction such that $\calF|_{\rho}$ does not have $w$-partial depth-$t$ decision tree. Then there exists an $(r,t)$-global witness for $\rho$.
\end{lem}

\begin{proof}
    Same as the proof of Corollary 1 in \cite{Lyu}.
\end{proof}

We now define \emph{partial global witnesses}, as there are far too many global witnesses to union bound over.

\begin{defn}
Let $t,w$ be two integers. Consider a list of $\g(k)\circ\{\AND,\OR\}_w$ circuits $\calF = \{F_1, \dots, F_m \}$. Suppose $\rho \in \{0,1,\star\}^n$ is a restriction. Let $(R, L_i, S_i, P_i, \beta_i)$ be a tuple, where
\begin{itemize}
    \item $1\le R \le \frac{t}{r}$ is an integer;
    \item $1\le L_1 \le L_2 \le \dots \le L_R\le m$ is a list of $R$ non-decreasing indices;
    \item $S_1.\dots, S_R$ is a list of $R$ integers such that $\sum_{i=1}^R S_i \in [t, t+w]$;
    \item $P_1, \dots, P_R$ is a list of partial witnesses. For every $i\in [R]$, $P_i$ has size $S_i$.
    \item $\beta_1,\dots, \beta_R$ are $R$ strings where $|\beta_i| = S_i$ for every $i\in [R]$.
\end{itemize}
We call $(R, L_i, S_i, P_i, \beta_i)$ a $(r, t)$-global partial witness for $\rho$, if we can complete $P_i$ to get a witness $W_i$ for every $i\in [R]$, such that $(R, L_i, S_i, W_i, \beta_i)$ is a global witness for $\rho$.
\end{defn}

By a simple induction, one can show the following claim.

\begin{claim}
\label{claim:uniquecomp}
    Given a global partial witness for $\rho$, there is exactly one way to complete it and get a global witness for $\rho$.
\end{claim}

We now construct a global witness searcher that will reconstruct a global witness from a partial one using advice.

    \begin{algorithm}[hbt!]
    \caption{Global Witness Searcher}
    \label{algo:global-searcher}
    \DontPrintSemicolon
    \kwInput{A list of DNFs ${\cal F} = \{F_1,\dots, F_m\}$, a global partial witness $(R,L_i,S_i,P_i,\beta_i)$, and an advice $z\in \{0,1\}^n$.}

    \init{
        $c \gets 1$. \;
        $\rho^{(1)} \gets \rho$. \;
    }
    \While{$c\le R$} {
        Run \Cref{alg:two} on $(F_{L_c}, \rho^{(c)}, P_c, y)$. If it reports ERROR, report ERROR and terminate the procedure. Otherwise let $W_c$ be the witness returned. \;
        $I_c \gets $ the set of variables involved in $W_c$. \;
        Identify $\beta_c$ as a partial assignment, where only $\beta_{I_c}$ is fixed. \;
        $\rho^{(c+1)}\gets \rho^{(c)} \circ \beta_c$. \;
        $c\gets c + 1$. \;
    }
    \textbf{return} $x$ \;
\end{algorithm}

We note the following important lemma about the searcher, which we denote $\calS$.

\begin{lem}\label{lemma:find-global-witness}
Let $P$ be a size $S$ global partial witness. Say $\rho$ is \emph{good} if $P$ is a global partial witness for $\rho$, then \[\Pr_z[\calS(z,P)\text{ is a global witness for }\rho(\Lambda, z) |\rho\text{ is good}]=2^{-S}\]  
\end{lem}
\begin{proof}
  Same as proof of Lemma 8 in \cite{Lyu}, but we instead appeal to \Cref{lem:search} whenever Lyu's proof refers to Lemma 6.   
\end{proof}

By this lemma, we have \begin{align}\Pr_{\Lambda, z}[\rho(\Lambda, z)\text{ is good for }P] & = \E_\Lambda\frac{\Pr_z[\calS(z,P)\text{ is a global witness for }\rho(\Lambda, z)]}{\Pr_z[\calS(z,P)\text{ is a global witness for }\rho(\Lambda, z) |\rho\text{ is good}]} \nonumber\\ & = 2^S\E_\Lambda\Pr_z[\calS(z,P)\text{ is a global witness for }\rho(\Lambda, z)] \nonumber \\ & \le 2^S\E_z\Pr_\Lambda[\calS(z,P)\text{ is a global witness for }\rho(\Lambda, z)] \nonumber\\ 
 & \le (2p)^s \label{eqn:ub} \end{align}
 where the last inequality follows from the fact that $\rho$ needs to keep the $S$ variables specified by the global witness alive in order to have any hope of being witnessed by it. By the $(t+w)$-wise $p$-boundedness of $\Lambda$, this happens with probability $\le p^S$.\

 Finally, if we let $N_S$ be the number of global witnesses of size $S$, we can see by using \Cref{lem:globw} and \Cref{claim:uniquecomp}, along with (\ref{eqn:ub}) that
 \begin{align}\Pr_\rho[\calF|_\rho\text{ has no $r$-partial depth-$t$ $\dt$} ] & \le \sum_P \Pr_\rho[P\text{ is a global partial witness for }\rho] \nonumber\\ & \le \sum_{S=t}^{t+w}N_S(2p)^S \label{eqn:final}\end{align}

 We can upper bound $N_S$ as follows.
 \begin{itemize}
     \item There are $\le \frac{t}r\cdot \binom{m}{t/r}\le 2m^{t/r}$ ways to pick $(R,L_i)$
     \item There are $\le 2^{S}$ choices for $(S_i)$ (since there are $\le 2^{n-1}$ ways to write $n$ as an ordered partition)
     \item From \Cref{thm:switch}, we know that there are $(8w)^S2^k$ partial witnesses of size $S$, giving a total amount of $\le \prod_i (8w)^{S_i}2^k = (8w)^S2^{kR}\le (8w)^S2^{kt/r}$
     \item There are clearly $2^{\sum S_i} = 2^S$ possibilities for $(\beta_i)$.
 \end{itemize}

 Combining all this tells us that $N_S\le 2m^{t/r}(32w)^S2^{kt/r}$. Hence, from (\ref{eqn:final}), we deduce
\begin{align*}\Pr_\rho[\calF|_\rho\text{ has no $r$-partial depth-$t$ $\dt$} ] & \le \sum_{S=t}^{t+w}N_S(2p)^S \le  m^{t/r}2^{kt/r}\sum_{S=t}^{t+w}(32pw)^S\le 4(64(2^km)^{1/r}pw)^t\label{eqn:final}\end{align*}

\end{proof}

\begin{thm}[Proof of \Cref{thm:demulti}]
\label{apx:demulti}
Let ${\cal F} = \{F_1,\dots, F_m\}$ be a list of size $m$ $\g(k)\circ \{\AND,\OR\}_w$ circuits.
     Let $(\Lambda, z)$ be a joint random variable such that 
     \begin{itemize}
         \item $\Lambda$ is a $(t+w)$-wise $p$-bounded subset of $[n]$
         \item Conditioned on any instance of $\Lambda$, $z$ $\eps$-fools CNF of size $\le m^2$.
     \end{itemize}
     Then  
     \[\Pr_{\Lambda, z}[{\cal F}|_{\rho(\Lambda, x)}\text{ has no $r$-partial depth-$t$ }\dt] \le 4(m2^k)^{t/r}(64pw)^t + (64wm)^{t+w}(2m)^{2kt/r}\cdot \eps\]
\end{thm}

\begin{proof}
    All steps of the proof of \Cref{lem:multiswitch} follow identically until we reach the step \begin{align*}\Pr_{\Lambda, z}[F|_{\rho(\Lambda, x)}&\text{ has no $w$-partial depth-$t$ }\dt] \\ & = \Pr\sum_{(R,L,S_i,P_i,\beta_i)}\Pr_{\Lambda, z}[(R,L,S_i,P_i,\beta_i)\text{ is a global partial witness for }\rho(\Lambda, z)]\end{align*} 
    From \Cref{claim:uniquecomp}, we can deduce the event decomposes
    \begin{align*}{\mathbbm 1}\{(R, L_i
, S_i
, P_i
, \beta_i)&\text{ is a global partial witness for }\rho(\Lambda, z)\}
 \\ & =\sum_{
W_i:\text{completion of }P_i}
{\mathbbm 1}\{(R, L_i
, S_i
, W_i
, \beta_i)\text{ is a global witness for }\rho(\Lambda, z)\}.\end{align*}

For a fixed $\Lambda$ and $(R, L_i, S_i, W_i, \beta_i)$, we can let \[h_\Lambda^{(R, L_i, S_i, W_i, \beta_i)}(z) = {\mathbbm 1}\{(R, L_i
, S_i
, W_i
, \beta_i)\text{ is a global witness for }\rho(\Lambda, z)\}.\]

We will now show that $h$ is a predicate computable by a small size CNF, so that $z$ will fool it. $h$ is true iff $W_i$ is a witness for $F_{L_i}|_{\rho(\Lambda, z)\circ\beta_1\dots \beta_{i-1}}$ for all $1\le i\le R$. Now for each $i$, one can verify $W_i = (r',\ell_i, s_i, B_i, \alpha_i)$ is a witness by a size $m$ CNF as follows.

\begin{itemize}
    \item For all $j<\ell_1$, $C_j$ is falsified by $\rho(\Lambda, z)$. This is true iff $(\neg C_1)\wedge (\neg C_2)\wedge \dots\wedge (\neg C_{\ell_1})$. Notice $\neg C_i$ becomes an $\OR$ clause by De Morgan's Law
    \item $C_{j_1}$ is satisfied, which is an $\AND$ of variables.
    \item $C_j$ is falsified by $\rho(\Lambda, z)\circ \alpha_1$ for $\ell_1 < j < \ell_2$, which is true iff $(\neg C_{\ell_1+1})\wedge \dots\wedge (\neg C_{\ell_2-1})$. Each $\neg C_i$ is an $\OR$ clause
    \item and so on and so forth until we verify $C_{j_{r'}}$.
\end{itemize} 

Each bullet points gives a disjunction of (maybe trivial) conjunctions, so for all bullet points to hold, we simply take the $\AND$ of all of them, resulting in a CNF whose size is bounded by $m$ (since our CNF is essentially $F_{L_i}$ but with some gates and negations changed. Hence, if we want to verify $W_i$ simultaneously over all $i$, we take the $\AND$ of all $R\le m$ CNFs to get a size $m^2$ CNF. Hence, $z$ $\eps$-fools $h$. From \Cref{thm:switch}, we know over a uniform string $x$, \[\sum_{
(R, L_i, S_i, W_i, \beta_i)}
\E_{\Lambda} \E_x h_\Lambda^{(R, L_i, S_i, W_i, \beta_i)}(x)\le 4(m2^k)^{t/r}(64pw)^t\]

Therefore, over $z$, we have 

\begin{align*}\Pr_{\Lambda, z}[F|_{\rho(\Lambda, x)}\text{ has no $w$-partial} \text{ depth-$t$ }\dt] & =  \sum_{
(R, L_i, S_i, W_i, \beta_i)}
\E_{\Lambda} \E_z h_\Lambda^{(R, L_i, S_i, W_i, \beta_i)}(z) \\ & \le \sum_{
(R, L_i, S_i, W_i, \beta_i)}
\E_{\Lambda} (\eps + \E_x h_\Lambda^{(R, L_i, S_i, W_i, \beta_i)}(x)) \\ & \le 4(m2^k)^{t/r}(64pw)^t + \sum_{(R, L_i, S_i, W_i, \beta_i)}\eps.  \end{align*}
The number of tuples $(R, L_i, S_i, W_i, \beta_i)$ can be bounded as follows. 
\begin{itemize}
\item From the proof of \Cref{apx:multiswitch}, we know there are $2m^{t/r}2^{kt/r}(32w)^S$ many global partial witnesses of size $S$.
    \item We now multiply by the number of $(\ell_i)$ possible for each partial $P_i$ of size $S_i$, which is at most $\binom{m}{t+k} < m^{S_i+k}$. Hence, the total number of $(W_i)$ over all $1\le i\le R$ is \[\prod_{i}m^{S_i+k} = m^{S+kR}\le m^{S+ct/r}\]
\end{itemize}
Hence the total number of size-$S$ tuples is upper bounded by \[2m^{t/r}2^{kt/r}(32w)^Sm^{S+kt/r}\le (32mw)^{S}(2m)^{2kt/r}.\] Summing over all $t\le S\le t+w$ gives us a grand total of \[\sum_{S=t}^{t+w}(32mw)^{S}(2m)^{2kt/r}\le (64mw)^{S}(2m)^{2kt/r}.\] Therefore, 

\[\Pr_{\Lambda, z}[F|_{\rho(\Lambda, x)}\text{ has no $r$-partial} \text{ depth-$t$ }\dt]\le 4(m2^k)^{t/r}(64pw)^t + (64mw)^{t+w}(2m)^{2kt/r}\cdot \eps. \]
\end{proof}

\end{document}